\documentclass[a4paper,11pt]{article}
\usepackage{amsmath,amsfonts,amssymb,amsthm}
\usepackage{mathtools}
\usepackage{latexsym}
\usepackage{mathrsfs}
\usepackage{color}
\usepackage{graphicx} 
\usepackage[colorlinks,linkcolor=magenta,citecolor=cyan]{hyperref}
\usepackage[
    backend=bibtex,
    style=alphabetic,
    alldates=year,
    isbn=false,
    eprint=true,
    doi=true,
    giveninits=true,
    maxcitenames=4,
    maxbibnames=10,
    block=ragged
    ]{biblatex}
\newbibmacro*{bbx:parunit}{%
  \ifbibliography
    {\setunit{\bibpagerefpunct}\newblock
     \usebibmacro{pageref}%
     \clearlist{pageref}%
     \setunit{\adddot\par\nobreak}}
    {}}

\renewbibmacro*{doi+eprint+url}{%
  \usebibmacro{bbx:parunit}
  \iftoggle{bbx:doi}
    {\printfield{doi}}
    {}%
  \iftoggle{bbx:eprint}
    {\usebibmacro{eprint}}
    {}%
  \iftoggle{bbx:url}
    {\usebibmacro{url+urldate}}
    {}}

\renewbibmacro*{eprint}{%
  \usebibmacro{bbx:parunit}
  \iffieldundef{eprinttype}
    {\printfield{eprint}}
    {\printfield[eprint:\strfield{eprinttype}]{eprint}}}

\renewbibmacro*{url+urldate}{%
  \usebibmacro{bbx:parunit}
  \printfield{url}%
  \iffieldundef{urlyear}
    {}
    {\setunit*{\addspace}%
     \printtext[urldate]{\printurldate}}}  
\DeclareFieldFormat{journaltitle}{\mkbibemph{#1},} 
\DeclareFieldFormat[inbook,thesis]{title}{\mkbibemph{#1}\addperiod} 
\DeclareFieldFormat[article]{title}{#1} 
\DeclareFieldFormat[article]{volume}{\textbf{#1}\addcolon\space}
\DeclareFieldFormat{pages}{#1}

\usepackage{enumerate}
\usepackage{fancyhdr}

\usepackage{todonotes,soul}
\setstcolor{red}

\iffalse 
  \let\realtitle\title
  
  \renewenvironment{abstract}{\begin{realabstract}\textbf{This is a draft of the summarizing publication for the doctoral thesis. Please do not make it publicly available.}\\}{\end{realabstract}}
  \renewcommand{\title}[1]{\realtitle{\textbf{Draft:} #1}}
  
  \renewcommand{\st}[1]{}
  \renewcommand{\todo}[1]{}
  \usepackage[left=3cm,right=3cm,top=2cm,bottom=2cm,marginparwidth=4cm]{geometry}
  \usepackage[scale=1.5,firstpageonly]{draftwatermark}  
\else
 \usepackage[left=3.5cm,right=3.5cm,top=2cm,bottom=2cm,marginparwidth=4cm]{geometry}

\usepackage{bbold}				
\usepackage{bbm}
\usepackage{enumitem}           
\usepackage{braket}             

\usepackage{mathtools}          
\usepackage{slashed}

\usepackage{hyphenat}
\theoremstyle{plain}
\newtheorem{thm}{Theorem}[section]
\newtheorem{defn}[thm]{Definition}
\theoremstyle{remark}
\newtheorem{rem}[thm]{Remark}
\theoremstyle{plain}
\newtheorem{prop}[thm]{Proposition}
\newtheorem{cor}[thm]{Corollary}
\newtheorem{lem}[thm]{Lemma}

\newtheorem{ex}[thm]{Example}

\DeclareMathOperator{\tr}{tr}

\DeclareMathOperator{\ch}{ch}
\DeclareMathOperator{\sh}{sh}
\let\th\relax
\DeclareMathOperator{\th}{th}

\usepackage{bm}                 
\newcommand{\bzeta}{\boldsymbol{\zeta}}
\newcommand{\btheta}{\boldsymbol{\theta}}

\newcommand{\balpha}{\boldsymbol{\alpha}}
\newcommand{\bbeta}{\boldsymbol{\beta}}

\newcommand{\bpi}{\boldsymbol{\pi}}

\newcommand{\bkappa}{\boldsymbol{\kappa}}

\newcommand{\lParent}{(}
\newcommand{\rParent}{)}

\newcommand{\non}{\nonumber}

\newcommand{\normord}[1]{{\vcentcolon}\!\mathrel{#1}\!{\vcentcolon}}

\newcommand{\email}[1]{\mbox{\href{mailto:#1}{#1}}}

\newcommand{\Mdiag}{M} 
\newcommand{\Mdiagtwo}{M^{\otimes 2}} 
\newcommand{\Mdiagthree}{M^3} 

\numberwithin{equation}{section}

\newcommand{\orcid}[1]{\small \href{https://orcid.org/#1}{\textcolor[HTML]{A6CE39}{\aiOrcid}}}

\title{Quantum energy inequalities in integrable models with several particle species and bound states}
\author{
Henning Bostelmann\thanks{University of York, Department of Mathematics, York YO10 5DD, United Kingdom. \newline Current address: Merseburg University of Applied Sciences, Department of Engineering and Natural Sciences, Eberhard-Leibnitz-Stra\ss e 2, 06217 Merseburg, Germany. \newline E-mail: \email{henning.bostelmann@hs-merseburg.de}}
\and 
Daniela Cadamuro\thanks{University of Leipzig, Institute for Theoretical Physics, Leipzig, Germany. \newline E-mail: \email{cadamuro@itp.uni-leipzig.de}}
\and
Jan Mandrysch\thanks{University of Leipzig, Institute for Theoretical Physics, Leipzig, Germany. \newline E-mail: \email{jan.mandrysch@fau.de}}
}

\begin{document}

\maketitle

\begin{abstract}
    We investigate lower bounds to the time-smeared energy density, so-called quantum energy inequalities (QEI), in the class of integrable models of quantum field theory. Our main results are a state-independent QEI for models with constant scattering function and a QEI at one-particle level for generic models. In the latter case, we classify the possible form of the stress-energy tensor from first principles and establish a link between the existence of QEIs and the large-rapidity asymptotics of the two-particle form factor of the energy density. Concrete examples include the Bullough-Dodd, the Federbush, and the $O(n)$-nonlinear sigma models.
\end{abstract}

\section{Introduction}

It is well known that the energy operator in quantum field theory (QFT) is positive, while the energy density $T^{00}$ may be locally negative. However, for physically reasonable theories, bounds on this negativity are expected when local averages are taken: quantum energy inequalities (QEIs).  They may, for example, take the form
\begin{equation}\label{eq:introQEI}
    \langle \varphi, \int dt g(t)^2 T^{00}(t,x) \varphi \rangle \geq - c_g \| \varphi\|^2,
\end{equation}
where the constant $c_g$ does not depend on $\varphi$, and the inequality holds for a suitably large set of vectors $\varphi$. In Minkowski space, the bound is also uniform in $x$.

Without these bounds, accumulation of negative energy might lead to violations of the second law of thermodynamics \cite{For78}. They also have significant importance in semiclassical gravity, where the expectation value of $T^{\mu\nu}$ appears on the right-hand side of the Einstein equations. In this context, QEIs can yield constraints on exotic spacetime geometries and lead to generalized singularity theorems extended from classical results in general relativity; see \cite[Sec.~5]{KS20} for a review.

QEIs have been established quite generically in linear QFTs, including QFTs on curved spacetimes; see \cite{Few12} for a review. They are also known in 1+1-dimensional conformal QFTs \cite{FH05}. However, their status is less clear in self-interacting models, i.e., models with a nontrivial scattering matrix between particles. Some generic results, weaker than \eqref{eq:introQEI}, can be obtained from operator product expansions \cite{BF09}. Concrete results in models with self-interaction are rare, though.

The situation is somewhat better in 1+1-dimensional integrable models. In these models, the scattering matrix is constrained to be factorizing but nonetheless allows for a large class of interactions; see, e.g., \cite{KW78,ZZ79,AAR01}. A QEI in this context was first established in the Ising model \cite{BCF13}. Also, a QEI at one-particle level (i.e., where \eqref{eq:introQEI} holds for one-particle states $\varphi$) has been obtained more generally for models with one scalar particle type and no bound states \cite{BC16}.

The class of integrable models is much richer, though---they can also describe several particle species with a more complicated scattering matrix between them or particles with inner degrees of freedom; further, these particles may form bound states\footnote{Bound states are understood as poles of the scattering matrix within the so-called physical strip. See \cite[Sec.~2.2]{BFKZ99} for further details.}. This article aims to generalize the results of \cite{BCF13,BC16} to these cases. 

As an a priori problem, one may ask what form the energy density operator $T^{00}$ takes in these models, even at one-particle level. The classical Lagrangian is often used as heuristic guidance; however, if one takes an inverse scattering approach to integrable models, starting by prescribing the two-particle scattering function, then a classical Lagrangian may not even be available in all cases. Instead, we will restrict the possible form of the energy density starting from generic physical assumptions (such as the continuity equation, but initially disregarding QEIs); see Theorem~\ref{thm:tformgen} below. 

We then ask whether QEIs can hold for these energy densities. 
There are two main results: First, we confine ourselves to a class of models with \emph{rapidity-independent} scattering function, i.e., where the scattering matrix is independent of the particle momenta. In this setup, for a canonical choice of energy density, we establish a QEI in states of arbitrary particle number (Theorem~\ref{thm:qeiconstS}). Second, for generic scattering functions, we give necessary and sufficient criteria for QEIs to hold \emph{at one-particle level} (Theorem~\ref{thm:qeimain}). Here it turns out that the existence of QEIs critically depends on the large-rapidity behaviour of the two-particle form factor $F_2$ of the energy density.

We apply our results to several concrete examples, namely, to the Bullough-Dodd model (Sec.~\ref{sec:bd}) which has bound states, to the Federbush model (Sec.~\ref{sec:fb}) as an interacting model with rapidity-independent scattering function, and to the $O(n)$ nonlinear sigma model (Sec.~\ref{sec:onsigma}) which features several particle species. In particular, we investigate how QEIs further restrict the choice of the stress-energy tensor in these models, sometimes fixing it uniquely.

In short, the remainder of this article is organized as follows.
We recall some background on integrable QFTs in Section~\ref{sec:prelim} and discuss the possible form of the energy density in Section~\ref{sec:stform}.
Section~\ref{sec:constantS} establishes a QEI in models with constant scattering function, and Section~\ref{sec:onepQEI} for more generic scattering functions but only at one-particle level. For controlling the large-rapidity asymptotics of $F_2$, critically important to our results in Section~\ref{sec:onepQEI}, we first explain the relation between the scattering function  and the so-called ``minimal solution'' in Section~\ref{sec:minimal}, with technical details given in the appendix (which contains known facts as well as original results). This is then applied to examples in Section~\ref{sec:examples}. Conclusion and outlook follow in Section~\ref{sec:conclusion}.

This article is based on the PhD thesis of one of the authors \cite{Man23}.

\newpage

\section{Preliminaries}\label{sec:prelim}

\subsection{General notation}\label{subsec:generalnotation}

We will work on 1+1-dimensional Minkowski space $\mathbb{M}$. The Minkowski metric $g$ is conventionally chosen to be $\operatorname{diag} (+1,-1)$ and the Minkowski inner product will be denoted by $p.x = g_{\mu\nu} p^\mu x^\nu$. A single parameter, called \emph{rapidity}, conveniently parametrizes the mass shell on $\mathbb{M}$. In this parameterization, the momentum at rapidity $\theta$ is given by $p^0(\theta;m) := m \ch \theta$ and $p^1(\theta;m):= m \sh \theta$, where $m>0$ denotes the mass. We will use $\theta,\eta,\lambda$ to denote real and $\zeta$ to denote complex rapidities. Introducing the open and closed strips, $\mathbb{S}(a,b):= \mathbb{R}+i(a,b)$ and $\mathbb{S}[a,b]:=\mathbb{R}+i[a,b]$, respectively, the region $\mathbb{S}[0,\pi]$ will be of particular significance and is referred to as the \emph{physical strip}.

In the following, let $\mathcal{K}$ be a finite-dimensional complex Hilbert space with inner product $(\cdot , \cdot)$, linear in the second position. We denote its extension to $\mathcal{K}^{\otimes 2}$ as $(\cdot,\cdot)_{\mathcal{K}^{\otimes 2}}$ and the induced norm as $\lVert \cdot \rVert_{\mathcal{K}^{\otimes 2}}$; i.e., for $v_i,w_i\in\mathcal{K}$, $i=1,2$ we have $(v_1\otimes v_2, w_1 \otimes w_2)_{\mathcal{K}^{\otimes 2}} = (v_1,w_1)(v_2,w_2)$. 
For computations, it will be convenient to choose an orthonormal basis $\{ e_{\alpha} \}, \alpha\in\{1, \ldots ,\dim \mathcal{K}\}$. In this basis, we denote $v\in \mathcal{K}^{\otimes m}$ and $w\in\mathcal{B}(\mathcal{K}^{\otimes m},\mathcal{K}^{\otimes n})$ in vector and tensor notation by

\begin{equation}v^{\balpha} := (e_{\balpha},v), \quad w^{\balpha}_{\bbeta} := (e_{\balpha}, w e_{\bbeta}).\end{equation}
Operators on $\mathcal{K}$ or $\mathcal{K}^{\otimes 2}$ will be denoted by uppercase Latin letters. This also applies to vectors in $\mathcal{K}^{\otimes 2}$, which are identified with operators on $\mathcal{K}$ as follows: For an antilinear involution $J \in \mathcal{B}(\mathcal{K})$ (to be fixed later), the map $A \mapsto \hat{A}$ defined by
\begin{equation}\label{eq:hatnotation}
    \forall u,v\in\mathcal{K}: \quad (u, \hat{A} v) := (u \otimes Jv, A)_{\mathcal{K}^{\otimes 2}}
\end{equation}
    yields a vector space isomorphism between $\mathcal{K}^{\otimes 2}$ and $\mathcal{B}(\mathcal{K})$.
In particular, we consider the special element $I_{\otimes 2} \in \mathcal{K}^{\otimes 2}$ defined by $\widehat{I_{\otimes 2}} = \mathbbm{1}_\mathcal{K}$. For an arbitrary orthonormal basis $\{ e_\alpha\}_\alpha$ of $\mathcal{K}$, it is explicitly given by 
\begin{equation}
    I_{\otimes 2} = \sum_{\alpha} e_\alpha \otimes J e_\alpha.
\end{equation}

\begin{rem}\label{rem:identityelement}
    $I_{\otimes 2}$ is invariant under the action of $U^{\otimes 2}$ for any $U \in \mathcal{B}(\mathcal{K})$ with $U$ unitary or anti-unitary and $[U,J]=0$.
\end{rem}

\subsection{One-particle space and scattering function}

\begin{defn}\label{defn:1partlittlespace}
    A \emph{one-particle little space (with a global symmetry)} $(\mathcal{K},V,J,M)$ is given by a finite-dimensional Hilbert space $\mathcal{K}$, a unitary representation $V$ of a compact Lie group $\mathcal{G}$ on $\mathcal{K}$, an antiunitary involution $J$ on $\mathcal{K}$, and a linear operator $M$ on $\mathcal{K}$ with strictly positive spectrum. We further assume that $M$, $V(g)$ and $J$ commute with each other.
\end{defn}

Given such a little space $(\mathcal{K},V,J,M)$, we define the \emph{one-particle space} $\mathcal{H}_1 := L^2(\mathbb{R},\mathcal{K}) \cong L^2(\mathbb{R}) \otimes \mathcal{K}$,
on which we consider the (anti-)unitary operators, $\varphi \in \mathcal{H}_1$,
\begin{align}\label{eq:poincarerep}
    (U_1(x,\lambda) \varphi)(\theta) & := e^{ip(\theta;M).x} \varphi(\theta-\lambda),\quad (x,\lambda) \in \mathcal{P}_+^\uparrow \\
    \label{eq:u1j}
    (U_1(j) \varphi)(\theta) & := J\varphi(\theta), \\ 
    \label{eq:v1}
    (V_1(g)\varphi)(\theta) & := V(g)\varphi(\theta), \quad g \in \mathcal{G}.
\end{align}

This defines a unitary strongly continuous representation of the proper Poincare group $\mathcal{P}_+$ and of $\mathcal{G}$, where the antiunitary $U_1(j)$ is the PCT operator, representing spacetime reflection.

We will denote the spectrum of the mass operator $M$ as $\mathfrak{M} \subset (0,\infty)$ and its spectral projections as $E_m, m\in \mathfrak{M}$. Moreover, introduce the \emph{total energy-momentum operator $P^\mu$} on $\mathcal{H}_1^{\otimes 2}$ by
\begin{equation}
(P^\mu \varphi)(\btheta) := P^\mu(\btheta)\varphi(\btheta), \quad P^{\mu}(\theta_1,\theta_2) := p^\mu(\theta_1;M)\otimes \mathbbm{1}_\mathcal{K} + \mathbbm{1}_\mathcal{K} \otimes p^\mu(\theta_2;M), \quad \varphi \in \mathcal{H}_1^{\otimes 2}, 
\end{equation}
as well as the \emph{flip operator} $\mathbb{F} \in \mathcal{B}(\mathcal{K}^{\otimes 2})$ given by $\mathbb{F}(u_1\otimes u_2) = u_2 \otimes u_1$ ($u_{1,2} \in \mathcal{K}$).

\begin{defn}\label{defn:smatrix}
    Let $(\mathcal{K},V,J,M)$ be a one-particle little space. A meromorphic function \makebox{$S:\mathbb{C} \to \mathcal{B}(\mathcal{K}^{\otimes 2})$} with no poles on the real line is called \emph{S-function} iff for all $\zeta,\zeta'\in\mathbb{C}$ the following holds:
    
    \begin{enumerate}[label=(S\arabic*),ref=(S\arabic*)]
     \item Unitarity:\label{sunit} \quad $S(\bar\zeta)^\dagger = S(\zeta)^{-1}.$
     \item Hermitian analyticity:\label{sherm} \quad $S(\zeta)^{-1} = S(-\zeta).$
     \item CPT invariance:\label{scpt} \quad $J^{\otimes 2} \mathbb{F} S(\zeta) \mathbb{F} J^{\otimes 2} = S(\zeta)^\dagger.$
     \item Yang-Baxter equation:\label{sybeq} \\$(S(\zeta) \otimes \mathbbm{1}_{\mathcal{K}}) (\mathbbm{1}_{\mathcal{K}} \otimes S(\zeta+\zeta')) (S(\zeta') \otimes \mathbbm{1}_{\mathcal{K}}) = (\mathbbm{1}_{\mathcal{K}} \otimes S(\zeta')) (S(\zeta+\zeta') \otimes \mathbbm{1}_{\mathcal{K}})(\mathbbm{1}_{\mathcal{K}} \otimes S(\zeta)).$
     \item Crossing symmetry:\label{scross} \\ $\forall \, u_i,v_i \in \mathcal{K}, i=1,2:$   \; $(u_1 \otimes u_2, S(i\pi-\zeta) \, v_1 \otimes v_2)_{\mathcal{K}^{\otimes 2}} = (J v_1 \otimes u_1, S(\zeta) \, v_2 \otimes J u_2)_{\mathcal{K}^{\otimes 2}}. $
     \item Translational invariance:\label{strans} \\ $(E_m\otimes E_{m'})S(\zeta) = S(\zeta) (E_{m'}\otimes E_m), \quad m,m'\in\mathfrak{M}.$
     \item $\mathcal{G}$ invariance:\label{sinv} \\ $\forall \, g \in \mathcal{G}: \quad [S(\zeta),V(g)^{\otimes 2}] = 0.$
    \end{enumerate}
    An S-function is called \emph{regular} iff
    \begin{enumerate}[label=(S\arabic*),ref=(S\arabic*)]
     \setcounter{enumi}{7}
     \item Regularity:\label{sreg} \quad $\exists \kappa > 0: \quad S\restriction_{\mathbb{S}(-\kappa,\kappa)} \text{ is analytic and bounded.}$
        In this case, $\kappa(S)$ denotes the supremum of such $\kappa$'s.
    \end{enumerate}
\end{defn}

\begin{rem}\label{rem:scattfct}
    The S-function is the central object to define the interaction of the model. It is also referred to as auxiliary scattering function \cite[Eq.~\lParent 2.7\rParent]{BFKZ99} and closely related to the two-by-two-particle scattering matrix of the model, differing from it only by a ``statistics factor'', namely $-1$ on a product state of two fermions and $+1$ on fermion-boson- or boson-boson-vectors. The full scattering matrix is given as a product of two-by-two-particle scattering matrices of all participating combinations of one-particle states; see, e.g., \cite[Sec.~2]{BFKZ99} and \cite[Secs.~5--6]{BC21}.
\end{rem}

\begin{rem}\label{rem:scattinbasis}
    In examples below, we will choose a basis of $\mathcal{K}$ such that $J$ is given by $(J v)^\alpha = \overline{v^{\bar\alpha}}$ for $v\in\mathcal{K}$; here $\alpha \mapsto \bar\alpha$ is an involutive permutation on $\{1, \ldots ,\dim \mathcal{K}\}$, i.e., $\overline{\overline{\alpha}} = \alpha$. Then the relations \ref{sunit}, \ref{sherm}, \ref{scross}, and \ref{scpt} amount to unitarity plus the following conditions:
    \begin{equation}
        S_{\alpha\beta}^{\gamma\delta}(\zeta) = \overline{S^{\alpha\beta}_{\gamma\delta}(-\bar\zeta)} = S_{\bar\delta\bar\gamma}^{\bar\beta\bar\alpha}(\zeta),  \quad S_{\alpha\beta}^{\gamma\delta}(i\pi-\zeta) = S_{\beta\bar\delta}^{\bar\alpha\gamma}(\zeta).
    \end{equation}
\end{rem}

\subsection{Integrable models, form factors, and the stress-energy tensor}\label{subsec:formfactors}

From the preceding data---one-particle little space $(\mathcal{K},V,J,M)$ and S-function $S$---it is well-known how to construct an integrable model of quantum field theory (inverse scattering approach). This can be done at the level of $n$-point functions of local fields \cite{Smi92,BFK08} or more rigorously in an operator algebraic setting, at least provided that $S$ is regular, analytic in the physical strip, and satisfies an intertwining property \cite{AL17}. We give a brief overview of the construction here, focussing only on aspects that will be relevant in the following. A detailed account can be found in \cite{Man23}.

The interacting state space $\mathcal{H}$, on which our local operators will act, is an $S$-symmetrized Fock space generated by $S$-twisted creators $z^\dagger$ and annihilators $z$ known as \emph{ZF operators} \cite{LM95,LS14}. They are defined as operator-valued distributions $h \mapsto z^\sharp(h)$, $h\in \mathcal{H}_1=L^2(\mathbb{R},\mathcal{K})$ with $z(h) := (z^\dagger(h))^\dagger$ and
\begin{equation}
    (z^\dagger(h)\Psi)_n := \sqrt{n} \operatorname{Symm}_S (h \otimes \Psi_{n-1}), \quad \Psi \in \mathcal{H}.
\end{equation}
Here $\Psi_n$ is the $n$-particle component of $\Psi$, and $\operatorname{Symm}_S$ denotes $S$-symmetrization: For $n=2$ (other cases will not be needed here) and a $\mathcal{K}^{\otimes 2}$-valued function $f$ in two arguments, it can be defined as
\begin{equation}
    \operatorname{Symm}_S f := \tfrac{1}{2} (1+S_\leftarrow)f, \qquad S_\leftarrow f(\zeta_1,\zeta_2) := S(\zeta_2-\zeta_1)f(\zeta_2,\zeta_1).
\end{equation}
Products of $z^\dagger$ and $z$ can be linearly extended to arguments in tensor powers of $\mathcal{H}_1$; for instance with $h_1,h_2 \in \mathcal{H}_1$ and $z_1,z_2 \in \{z,z^\dagger\}$, we have $z_1z_2 (h_1\otimes h_2) := z_1(h_1) z_2(h_2)$. Defining also $S^{i\pi}(\zeta) := S(i\pi+\zeta)$ and given arbitrary $h_1,h_2 \in \mathcal{H}_1$ the \emph{ZF algebra} relations amount to
\begin{align}
    z^\dagger z^\dagger ( (1 - S_\leftarrow) (h_1\otimes h_2) ) & = 0, \label{eq:algalt1}\\
    z  z ( J^{\otimes 2} (1- S_\leftarrow) (h_1\otimes h_2) ) & = 0, \\
    z  z^\dagger ( h_1\otimes h_2 ) - z^\dagger z ( (1\otimes J) S^{i\pi}_\leftarrow (Jh_1\otimes h_2) ) & = \braket{h_1,h_2}\mathbbm{1}.\label{eq:algalt3}
\end{align}

To define locality in our setup, it is helpful to introduce two auxillary fields,
\begin{equation}
    \Phi(f) = z^\dagger(f^+) + z(U_1(j)f^-), \qquad
    \Phi^\prime(f) = U(j)\Phi(U_1(j)f)U(j),
\end{equation}
for arbitrary $f \in \mathcal{S}(\mathbb{M},\mathcal{K})$ and where $f^\pm(\theta) := \tilde{f}(\pm p(\theta;M))$ and $U(j)$ implements the CPT transformation on all of $\mathcal{H}$. $\Phi(f)$ and $\Phi^\prime(f)$ may be understood as being localized in a left, resp., right wedge containing the support of $f$. An operator $A$ is then referred to as localized in some bounded spacetime region $O \subset \mathbb{M}$, given as the intersection of a left and a right wedge, if it is relatively local to $\Phi$ and $\Phi^\prime$; for more details on this, we refer to \cite[Sec.~2.4]{BC15}, \cite{Lec15} and references therein.

Now, any such local operator $A$ can be expanded into a series of the form
\begin{equation}\label{eq:arakiexpansion}
    A = \sum_{n=0}^\infty \mathcal{O}_n [F_n^{[A]}]
\end{equation}
(see~\cite{BC15} for the case $\operatorname{dim}\mathcal{K}=1$). Here the $F_n^{[A]}$ are meromorphic functions of $n$ variables depending linearly on $A$ which are known as the \emph{form factors} of $A$; they satisfy a number of well-known properties, the \emph{form factor equations} \cite{BFK08}. In line with the literature, we will call $F_n$ the \emph{n-particle} form factor, though note that expectation values in $n$-particle \emph{states} generically have contributions from all zero- to $2n$-particle form factors. The symbols $\mathcal{O}_n$ are given by
\begin{align}\label{eq:of0}
    \mathcal{O}_0[F_0] & = F_0\mathbbm{1}, \\
    \mathcal{O}_1[F_1] & = z^\dagger(F_1) + z(JF_1(\cdot+i\pi)),\\
    \mathcal{O}_2[F_2] & = \frac{1}{2} z^\dagger z^\dagger (F_2) + z^\dagger z ( (1 \otimes J) F_2(\cdot, \cdot + i \pi)) + \frac{1}{2} zz (J^{\otimes 2} F_2(\cdot+i\pi, \cdot + i\pi )),  \label{eq:of2}  
\end{align}
and analogously for higher $n$, but only $n \leq 2$ will be needed in the following. Conversely, given $F_n$ that fulfil the form factor equations and suitable regularity conditions, \eqref{eq:arakiexpansion} defines a local operator $A$. The series \eqref{eq:arakiexpansion} is to be read in the sense of quadratic forms on $\mathcal{D} \times \mathcal{D}$ with a dense domain $\mathcal{D}\subset \mathcal{H}$, which we can take to consist of elements  $\Psi=(\Psi_n) \in \mathcal{H}$, where each $\Psi_n$ is smooth and compactly supported and $\Psi_n = 0$ for large enough $n$. With suitably chosen $F_n$, we can also regard each $\mathcal{O}_n[F_n]$ as an operator on $\mathcal{D}$, for example for $n=1$ if $F_1$ and $F_1(\cdot + i\pi)$ are square-integrable. 

In the following, we are interested in a specific local operator, \emph{the stress-energy tensor}, i.e., we study
\begin{equation}\label{eq:obs-edensity}
 A= T^{\mu\nu}(g^2) = \int dt \,g(t)^2 \, T^{\mu\nu}(t,0),
\end{equation}
averaged in time with a nonnegative test function $g^2$, $g\in\mathcal{S}_\mathbb{R}(\mathbb{R})$, and at $x^1=0$ without loss of generality; the integral is to be read weakly on $\mathcal{D}\times\mathcal{D}$. 
Also, we will focus on its two-particle coefficient $F_2^{[A]}$; this is because:

\begin{enumerate}[label=(\alph*)]
 \item In some models, the energy density has only these coefficients, i.e., $F_n^{[A]}=0$ for $n \neq 2$ (see Sec.~\ref{sec:constantS}).
 \item One-particle expectation values, which will partly be our focus, are determined solely by the coefficients $F_n^{[A]}$ for $n \leq 2$.
 \item The coefficients with $n < 2$ are not important for QEI results since the zero-point energy is expected to vanish ($F_{0}^{[A]} = 0$) and the coefficient $F_{1}^{[A]}$ yields only bounded contributions to the expectation values of $A$ (see Remark~\ref{rem:superpos01} below).
\end{enumerate}

\noindent
Under suitable regularity conditions, one has from \eqref{eq:obs-edensity},
\begin{equation} \label{eq:2partffstresstensor}
    F_2^{[A]}(\bzeta) = \int dt \, g(t)^2  F_2^{\mu\nu}(\bzeta;t,0), \quad \text{where } F_2^{\mu\nu}(\bzeta;x) := F_2^{[T^{\mu\nu}(x)]}(\bzeta).
\end{equation}
Assuming $F_0^{[A]} = 0$, only the $F_2^{[A]}$-term in \eqref{eq:of2} contributes to the one-particle expectation value of $A$. Using the definition of the ZF operators, \eqref{eq:arakiexpansion}, and \eqref{eq:2partffstresstensor}, the expectation value of the (time-smeared) stress-energy tensor in one-particle states $\varphi \in \mathcal{H}_1\cap\mathcal{D}$ evaluates to
\begin{equation}\label{eq:expvalues}
    \braket{\varphi, T^{\mu\nu}(g^2) \varphi} = \int d\theta d\eta \,dt\,  g(t)^2 \left( \varphi(\theta) , \widehat F_2^{\mu\nu}(\theta,\eta+i\pi;t,0) \varphi(\eta) \right)
\end{equation}
with $\widehat{F_2}$ as in Eq.~\eqref{eq:hatnotation} for each $\theta$, $\eta$. 

The above analysis applies to $T^{\mu\nu}$ under quite generic assumptions on the high-energy behaviour; for a detailed derivation we refer to \cite[Sec. 5.2]{Man23}. For our present purposes, details of the technical setup are not needed; in fact, we will proceed in the opposite way: We will select a suitable form factor $F_2^{\mu\nu}(\bzeta;x)$ for the stress-energy tensor, then use Eq.~\eqref{eq:expvalues} to \emph{define} $T^{\mu\nu}(g^2)$ as a quadratic form at one-particle level, i.e., on $\mathcal{H}_1\cap\mathcal{D}$, or more generally, the expansion~\eqref{eq:arakiexpansion} to define it for arbitrary particle numbers, as a quadratic form on $\mathcal{D}$.

\newpage
\section{The stress-energy tensor at one-particle level}\label{sec:stform}

This section analyses what form the stress-energy tensor $T^{\mu\nu}$ and, in particular, the energy density $T^{00}$ can take in our setup. 
Since our models do not necessarily arise from a classical Lagrangian, we study the stress-energy tensor using a ``bootstrap'' approach: We require a list of physically motivated properties for $T^{\mu\nu}$ and study which freedom of choice remains. 

Here we restrict our attention to the one-particle level, where the stress-energy tensor is determined by its form factor $F_2^{\mu\nu}$, as explained in Section~\ref{subsec:formfactors}. We will impose physically motivated axioms directly for the function $F_2^{\mu\nu}$; see properties (T1)--(T12) in Definition~\ref{def:set1} below. Without making claims on the existence of a full stress-energy tensor $T^{\mu\nu}$, we motivate these axioms by the expected features of $T^{\mu\nu}$ as follows:

First, $T^{\mu\nu}(x)$ should be a local field, i.e., commute with itself at spacelike separation. This property is well-studied in the form factor programme to integrable systems and is expected to be equivalent to the form factor equations \cite[Sec.~2]{Smi92}. The same relations can be justified rigorously in an operator algebraic approach, at least for a single scalar field ($\operatorname{dim} \mathcal{K}=1$) without bound states \cite{BC15}, with techniques that should apply as well for more general $\mathcal{K}$ \cite{Man23} and in the presence of bound states\footnote{K.~Shedid Attifa, work in progress}.
At one-particle level, where the form factor equations simplify, this yields properties \ref{tmero}--\ref{tper} below, with hermiticity of $T^{\mu\nu}(x)$ implying \ref{therm}; confer also \cite[Thm. 3.2.1, Prop. 5.2.2]{Man23}. The pole set $\mathfrak{P}$ appearing below is directly connected to the bound state poles of the S-function\cite{KW78,BFKZ99,BK02} and will be specified in the examples (Sec.~\ref{sec:examples}).

Further, $T^{\mu\nu}$ should behave covariantly under proper Poincar\'e transformations as a CPT-invariant symmetric 2-tensor \ref{tlorsym}, \ref{tpoincarecov}, \ref{tcpt}. It should be conserved, i.e., fulfil the continuity equation, \makebox{$\partial_\mu T^{\mu\nu} = 0$} \ref{tcont}, and integrate to the total energy-momentum operator, $P^\mu = \int T^{0\mu}(x^0, x^1)dx^1$ \ref{tdens}. Lastly, we demand that $T^{\mu\nu}$ is invariant under the action of $\mathcal{G}$ \ref{tinv} and, optionally, covariant under parity inversion \ref{tparity}.

For the following definition (and later on) we use the notation
$\bzeta = (\zeta_1,\zeta_2)$, $\overset{\leftarrow}{\bzeta} = (\zeta_2,\zeta_1)$, $\bpi = (\pi,\pi)$.
\begin{defn} \label{def:set1}
Given a little space $(\mathcal{K},V,J,M)$, an S-function $S$, and a subset $\mathfrak{P} \subset \mathbb{S}(0,\pi)$, a \emph{stress-energy tensor at one-particle level} (with poles $\mathfrak{P}$) is formed by functions $F_2^{\mu\nu}: \mathbb{C}^2\times \mathbb{M} \to \mathcal{K}^{\otimes 2},$ $\mu,\nu=0,1$, which for arbitrary $\bzeta \in \mathbb{C}^2$, $x\in\mathbb{M}$ satisfy

\begin{enumerate}[label=(T\arabic*),ref=(T\arabic*)]
 \item \label{tmero} Analyticity: \quad $F_2^{\mu\nu}(\zeta_1,\zeta_2;x)$ is meromorphic in $\zeta_2-\zeta_1$, where the poles within $\mathbb{S}(0,\pi)$ are all first-order and $\mathfrak{P}$ denotes the set of poles in that region.
 \item \label{treg} Regularity: \quad There exist constants $a,b,r \geq 0$ such that\\
 for all $|\Re (\zeta_2-\zeta_1)| \geq r$ and $\Im (\zeta_2-\zeta_1) \in [0,\pi]$ it holds that\\
 $\max_{\mu,\nu} ||F_2^{\mu\nu}(\zeta_1,\zeta_2;x)||_{\mathcal{K}^{\otimes 2}} \leq a \exp b \left( |\Re \zeta_1|+|\Re \zeta_2|\right).$
 \item \label{tssym} S-symmetry: \quad $F_2^{\mu\nu}(\bzeta;x) = S(\zeta_2-\zeta_1) F_2^{\mu\nu}(\overset{\leftarrow}{\bzeta};x).$
 \item \label{tper} S-periodicity: \quad $F_2^{\mu\nu}(\bzeta;x) = \mathbb{F} F_2^{\mu\nu}(\zeta_2,\zeta_1+i2\pi;x).$
  \item \label{therm} Hermiticity: \quad $F_2^{\mu\nu}(\bzeta;x) = \mathbb{F} J^{\otimes 2} F_2^{\mu\nu}(\overset{\leftarrow}{\bar\bzeta}+i\bpi;x).$
 \item \label{tlorsym} Lorentz symmetry: \quad $F_2^{\mu\nu} = F_2^{\nu\mu}.$
\item \label{tpoincarecov} Poincar\'e covariance: \quad For all $\lambda \in \mathbb{R}$ and $a \in \mathbb{M}$ it holds that $$\Lambda(\lambda)^{\otimes 2}F_2(\bzeta;\Lambda(\lambda)x+a) = e^{iP(\bzeta).a} F_2(\bzeta-(\lambda,\lambda);x), \quad \Lambda(\lambda) := \left( \begin{matrix} \ch(\lambda) & \sh (\lambda) \\ \sh(\lambda) & \ch(\lambda) \end{matrix} \right).$$
 \item \label{tcpt} CPT invariance: \quad $F_2^{\mu\nu}(\bzeta;x) = \mathbb{F} J^{\otimes 2} F_2^{\mu\nu}(\overset{\leftarrow}{\bar\bzeta};-x).$
 \item \label{tcont} Continuity equation: \quad $P_\mu(\bzeta) F_2^{\mu\nu}(\bzeta;x)=0.$
 \item \label{tdens} Normalization: $F_2^{0\mu}(\zeta,\zeta+i\pi;x) = \frac{M^{\otimes 2}}{2\pi} \mathcal{L}^{0\mu}(P(\zeta,\zeta+i\pi)) I_{\otimes 2}$ with
 \begin{equation}\label{eq:Lmunu}
\mathcal{L}^{\mu\nu}(p) := \frac{-p^\mu p^\nu + g^{\mu\nu}p^2}{p^2}.
 \end{equation}
 \item \label{tinv} $\mathcal{G}$ invariance: \quad $F_2^{\mu\nu}(\bzeta;x) = V(g)^{\otimes 2} F_2^{\mu\nu}(\bzeta;x),\quad g \in \mathcal{G}.$
\end{enumerate}
It is called \emph{parity-covariant} if, in addition,
\begin{enumerate}[label=(T\arabic*),ref=(T\arabic*)]
 \setcounter{enumi}{11}
 \item \label{tparity} Parity covariance
    $$F_2^{\mu\nu}(\bzeta;x^0,x^1) = \mathcal{P}^{\mu}_{\mu'} \mathcal{P}^{\nu}_{\nu'} F_2^{\mu'\nu'}(-\bzeta;x^0,-x^1), \quad \mathcal{P}^{\mu}_{\nu}= \left( \begin{matrix} 1 & 0 \\ 0 & -1\end{matrix} \right)^{\mu}_\nu.$$
\end{enumerate}
\end{defn}

\noindent
Property \ref{tpoincarecov} implies that for any $g \in \mathcal{S}(\mathbb{R})$,
\begin{equation}
    \int dt \, g^2(t) F_2^{\mu\nu}(\btheta;t,0) = \widetilde{g^2}(P_0(\btheta)) F_2^{\mu\nu}(\btheta;0) \quad \text{where}\; \widetilde{g^2}(p) = \int dt g(t)^2 e^{ipt}.
\end{equation}
Such $F_2^{\mu\nu}$ then defines the stress-energy tensor, $T^{\mu\nu}$, as a quadratic form between one-particle vectors by Eq.~\eqref{eq:expvalues}.
We are now in a position to characterize these one-particle stress-energy tensors.
\begin{thm}\label{thm:tformgen}
    Given a little space $(\mathcal{K}, V, J, M)$, an S-function $S$, and a subset $\mathfrak{P} \subset \mathbb{S}(0,\pi)$, then $F_2$ is a stress-energy tensor at one-particle level (with poles $\mathfrak{P}$) iff it is of the form
    \begin{equation}\label{eq:tformgen}
        F_2^{\mu\nu}(\zeta_1,\zeta_2;x) = \frac{M^{\otimes 2}}{2\pi} \mathcal{L}^{\mu\nu}(P(\bzeta)) e^{iP(\bzeta).x} F(\zeta_2-\zeta_1), \quad \bzeta=(\zeta_1,\zeta_2) \in \mathbb{C}^2,
    \end{equation}
    where $F:\mathbb{C}\to \mathcal{K}^{\otimes 2}$ is a meromorphic function which satisfies for all $\zeta \in \mathbb{C}$ that
    \begin{enumerate}[label=(\alph*),ref=(\alph*)]
        \item \label{f2poles} $F \restriction \mathbb{S}[0,\pi]$ has exactly the poles $\mathfrak{P}$;
        \item \label{f2bound} $\exists a,b,r >0\, \forall |\Re\zeta|\geq r: \quad \lVert F(\zeta) \rVert_{\mathcal{K}^{\otimes 2}} \leq a \exp( b |\Re \zeta|)$;
        \item \label{f2ssym} $F(\zeta) = S(\zeta) F(-\zeta)$;
        \item \label{f2period} $F(\zeta + i\pi) = \mathbb{F}F(-\zeta + i\pi) $;
        \item \label{f2cpt} $F(\zeta+i\pi) = J^{\otimes 2} F(\bar\zeta+i\pi)$;
        \item \label{f2ginv} $F = V(g)^{\otimes 2}F$ for all $g \in \mathcal{G}$;
        \item \label{f2norm} $F(i\pi) =I_{\otimes 2}$.
    \end{enumerate}
    It is parity covariant iff, in addition,
    \begin{enumerate}[label=(\alph*),ref=(\alph*)]
     \setcounter{enumi}{7}
     \item \label{f2parity} $F(\zeta+i\pi) = F(-\zeta+i\pi)\quad $ or, equivalently, $\quad F = \mathbb{F} F.$
    \end{enumerate}
\end{thm}

\begin{rem} 
    As can be seen from the proof, it is sufficient to require \ref{tdens} for $\mu=0$; the case $\mu=1$ is automatic.
\end{rem}

\begin{proof}[Proof of Theorem~\ref{thm:tformgen}]
    Assume $F_2$ to satisfy \ref{tmero}--\ref{tparity}. By Poincare covariance \ref{tpoincarecov}, it is given by
    \begin{equation}
        F_2(\bzeta;x) = e^{iP(\bzeta).x} \Lambda\left(-\tfrac{\zeta_1+\zeta_2}{2}\right)^{\otimes 2} F_2(-\tfrac{\zeta_2-\zeta_1}{2},\tfrac{\zeta_2-\zeta_1}{2};0).
    \end{equation}
    Define $G^{\mu\nu}(\zeta):= F_2^{\mu\nu}(-\tfrac{\zeta}{2},\tfrac{\zeta}{2};0)$ and observe that the conditions \ref{tmero} to \ref{tssym}, \ref{tcpt}, and \ref{tinv} imply that $G$ is meromorphic with pole set $\mathfrak{P}$ when restricted to $\mathbb{S}[0,\pi]$ and that for all $\mu,\nu=0,1$,
    \begin{equation}\begin{aligned}
        & \forall |\Re\zeta|\geq r:\,\lVert G^{\mu\nu}(\zeta) \rVert_{\mathcal{K}^{\otimes 2}} \leq a \exp(b|\Re \zeta|), & G^{\mu\nu}(\zeta) = S(\zeta) G^{\mu\nu}(-\zeta), \\
        & G^{\mu\nu}(\zeta+i\pi) = \mathbb{F}J^{\otimes 2} G^{\mu\nu}(-\bar\zeta+i\pi), & G^{\mu\nu}(\zeta) = V(g)^{\otimes 2} G^{\mu\nu}(\zeta).
    \end{aligned}\end{equation}
    Omit the Minkowski indices for the moment. Then combining \ref{therm} and \ref{tcpt} we obtain $F_2(\bzeta;x) = F_2(\bzeta+i\bpi;-x)$ and thus $G(\zeta) = G^\pi(\zeta)$, where $G^\pi(\zeta):= F_2(-\tfrac{\zeta}{2}+i\pi,\tfrac{\zeta}{2}+i\pi;0)$. Combining \ref{tper} with the preceding equality, we obtain $G(\zeta+i\pi) = \mathbb{F}G^\pi(-\zeta+i\pi)=\mathbb{F}G(-\zeta+i\pi)$. Moreover, by \ref{therm}, we have $G(\zeta+i\pi) = \mathbb{F}J^{\otimes 2} G(-\bar\zeta+i\pi) = J^{\otimes 2}G(\bar\zeta+i\pi)$. If we demand \ref{tparity}, this implies $G(\zeta+i\pi) = G(-\zeta+i\pi)$ and with the preceding properties also $G(\zeta) = \mathbb{F}G(\zeta)$. In summary, each $G^{\mu\nu}(\zeta), \, \mu,\nu=0,1$ satisfies properties \ref{f2poles}-\ref{f2ginv}, and possibly \ref{f2parity}, analogously.
    
    Due to the continuity equation \ref{tcont}, we have
    \begin{equation}(M_1+M_2)  G^{0\nu}(2\zeta)\ch \zeta + (M_1-M_2) G^{1\nu}(2\zeta)\sh\zeta  = 0, \quad \nu=0,1,\end{equation}
    where $M_1:= M \otimes \mathbbm{1}_\mathcal{K}$ and $M_2 := \mathbbm{1}_\mathcal{K} \otimes M$. Multiplying by the inverses of $M_1+M_2$ and $\ch \zeta$ (both are invertible), we find
    \begin{equation}G^{0\nu}(2\zeta) = \frac{-M_1+M_2}{M_1+M_2} G^{1\nu}(2\zeta) \th \zeta, \quad \nu=0,1.\end{equation}
    Defining $\operatorname{tr} G := g_{\mu\nu} G^{\mu\nu} = G^{00}-G^{11}$, we obtain
    \begin{equation}G^{\mu\nu}(\zeta) = \dfrac{1}{s(\zeta)^2-1} \left( \begin{matrix} s(\zeta)^2  & s(\zeta) \\ s(\zeta) & 1 \end{matrix} \right)^{\mu\nu} \operatorname{tr} G(\zeta)\end{equation}
    with $s(\zeta) := \tfrac{-M_1+M_2}{M_1+M_2} \th \tfrac{\zeta}{2}=\tfrac{P^1(-\zeta/2,\zeta/2)}{P^0(-\zeta/2,\zeta/2)}$. This yields
    \begin{equation}
        G^{\mu\nu}(\zeta) = \mathcal{L}^{\mu\nu}(P(-\tfrac{\zeta}{2},\tfrac{\zeta}{2})) \operatorname{tr} G(\zeta).
    \end{equation}
    On the other hand, from \ref{tdens} we infer
    $G^{00}(i\pi) = \frac{1}{2\pi} M^{\otimes 2}I_{\otimes 2} $;
    since $\mathcal{L}^{0 0}(P(-\tfrac{\zeta}{2},\tfrac{\zeta}{2})) \to \delta(M_1-M_2)$ as $\zeta\to i\pi$, this yields 
    \begin{equation}\operatorname{tr} G(i\pi) = \tfrac{M^{\otimes 2}}{2\pi} I_{\otimes 2}.\end{equation}
    Define now
    \begin{equation}F(\zeta):= \left( \tfrac{M^{\otimes 2}}{2\pi} \right)^{-1} \operatorname{tr} G(\zeta).\end{equation}
    Since $M^{\otimes 2}$ commutes with all $S(\zeta)$, $\mathbb{F}$, $J$ and $V(g)$, we find that $F$ satisfies properties \ref{f2poles}--\ref{f2norm}, plus \ref{f2parity} in the parity-covariant case.
    We have thus shown \eqref{eq:tformgen} for arguments of the form $(-\zeta/2,\zeta/2;x)$. That \eqref{eq:tformgen} holds everywhere now follows from \ref{tpoincarecov} together with the identity
    \begin{equation}
        \mathcal{L}^{\mu\nu}(P(\bzeta)) = \Lambda\left(-\tfrac{\zeta_1+\zeta_2}{2}\right)^{\mu}_{\mu'} \Lambda\left(-\tfrac{\zeta_1+\zeta_2}{2}\right)^{\nu}_{\nu'} \mathcal{L}^{\mu'\nu'}(P(-\tfrac{\zeta_2-\zeta_1}{2},\tfrac{\zeta_2-\zeta_1}{2})),
    \end{equation}
    which can be derived from the relation $p(\theta+\lambda;m) = \Lambda(\lambda) p(\theta;m)$.---%
    The converse direction, to show that \eqref{eq:tformgen} satisfies \ref{tmero} to \ref{tinv} (and \ref{tparity} provided that \ref{f2parity}) is straightforward.
\end{proof}
Let us call $X \in \mathcal{K}^{\otimes 2}$ \emph{diagonal in mass} if 
\begin{equation}\label{eq:diaginmass}
 (E_m \otimes E_{m'}) X = 0 \quad \text{for all } m\neq m'. 
\end{equation}
Equivalently, $\hat X$ commutes with $M$. On such $X$, all of $M_1$, $M_2$ and $(M\otimes M)^{1/2}$ act the same and in a slight abuse of notation we will use $\Mdiag$ to denote any of these. If $F$ has this property, i.e., $F(\zeta)$ has it for all $\zeta\in \mathbb{C}$, then the above result simplifies:
\begin{cor}\label{cor:diagmass}
    Assume that $F$ is diagonal in mass, or equivalently, that $\tr F_2({\cdot};x)$ is diagonal in mass for some $x$. Then
    $F_2^{\mu\nu}(\zeta_1,\zeta_2+i\pi;0) = G^{\mu\nu}_\mathrm{free}( \tfrac{\zeta_1+\zeta_2}{2}) F(\zeta_2-\zeta_1+i\pi)$ with
    \begin{equation}\label{eq:gfree}
         G_\mathrm{free}^{\mu\nu}(\zeta) := \frac{\Mdiagtwo}{2\pi} \left( \begin{matrix} \ch^2 \zeta & \sh \zeta \ch \zeta \\ \sh \zeta \ch \zeta & \sh^2 \zeta \end{matrix} \right)^{\mu\nu} .
    \end{equation}
    The energy density, in particular, becomes
    \begin{equation}\label{eq:f200diagmass}
    	F_2^{00}(\theta,\eta+i\pi;x) = \frac{\Mdiagtwo}{2\pi} \ch^2\left(\frac{\theta+\eta}{2}\right) e^{i(P(\theta)-P(\eta)).x} F(\eta-\theta+i\pi) .	
    \end{equation}
\end{cor}

\begin{proof}
	On $X \in \mathcal{K}^{\otimes 2}$ which is diagonal in mass we can simplify
	\begin{equation}
		P(\zeta_1,\zeta_2+i\pi)X = \big(p(\zeta_1;\Mdiag)-p(\zeta_2;\Mdiag)\big)X = 2 \Mdiag \sh \tfrac{\zeta_1-\zeta_2}{2} \left( \begin{matrix}\sh \tfrac{\zeta_1+\zeta_2}{2} \\ \ch \tfrac{\zeta_1+\zeta_2}{2}\end{matrix}\right) X.
	\end{equation}
	A straightforward computation shows that $\mathcal{L}^{\mu\nu}(P(\zeta_1,\zeta_2+i\pi))X$ depends only on $\frac{\zeta_1+\zeta_2}{2}$ and yields the proposed form of $F_2$.
\end{proof}

\begin{rem}\label{rem:f1}
    In some models, the one-particle form factor of the stress-energy tensor, $F_1$, is nonzero; in particular in models with bound states, where $F_1$ is linked to the residues of $F_2$ \cite[Sec.~3, Item d]{BFK08}. The general form of $F_1^{\mu\nu}(\zeta;x) := F_1^{[T^{\mu\nu}(x)]}(\zeta)$ can be determined analogous to Theorem~\ref{thm:tformgen}. In this case the continuity equation, $P_\mu(\zeta) F_1^{\mu\nu}(\zeta;x)$, implies that $F_1^{0\nu}(0;x) = 0$. Poincar\'e covariance yields that $F_1(\zeta;x) = e^{ip(\zeta;M).x} \Lambda(-\zeta)^{\otimes 2} F_1(0;0)$. As a result,
    \begin{equation}
        F_1^{\mu\nu}(\zeta;x) = e^{ip(\zeta;M).x} \left( \begin{matrix} \sh^2\zeta & -\sh \zeta \ch \zeta \\ -\sh \zeta\ch\zeta & \ch^2 \zeta\end{matrix}\right)  F_1(0),
    \end{equation}
    where $F_1(0)\in \mathcal{K}$ is constant. Hermiticity and $\mathcal{G}$-invariance imply $F_1(0) = JF_1(0) = V(g)F_1(0)$ for all $g\in \mathcal{G}$. The analogues of the other conditions in Theorem~\ref{thm:tformgen} are automatically satisfied. 
\end{rem}

It is instructive to specialize the above discussion to free models: For a single free particle species of mass $m$, either a spinless boson ($S= 1$) or a Majorana fermion ($S=-1$), we have $\mathcal{K}=\mathbb{C}$, $Jz = \bar{z}$, $M = m 1_\mathbb{C}$, $\mathcal{G}=\mathbb{Z}_2$, and $V(\pm 1) = \pm 1_\mathbb{C}$. The canonical expressions for the stress-energy tensor at one-particle level are
\begin{align}\label{eq:tfree}
    F^{\mu\nu}_{2,\mathrm{free},+}(\theta,\eta+i\pi;x) & = G_\mathrm{free}^{\mu\nu}\left( \tfrac{\theta+\eta}{2}\right)e^{i(p(\theta;m)-p(\eta;m)).x}, \\
    F^ {\mu\nu}_{2,\mathrm{free},-}(\theta,\eta+i\pi;x) & = \ch \tfrac{\theta-\eta}{2} F^{\mu\nu}_{2,\mathrm{free},+}(\theta,\eta+i\pi;x)
\end{align}
for the bosonic and the fermionic case, respectively; these conform to Definition~\ref{def:set1}, including parity covariance. Theorem~\ref{thm:tformgen} applies with $F_+(\zeta) = 1$ and $F_-(\zeta+i\pi) = \ch \tfrac{\zeta}{2}$. 
Moreover, note that $F_n^{[T^{\mu\nu}(x)]} = 0$ for $n \neq 2$ for these examples.

\section{A state-independent QEI for constant scattering functions}\label{sec:constantS}
In this section, we treat scattering functions $S$ which are constant, i.e., independent of rapidity. In this case, \ref{sunit} and \ref{sherm} imply that $S \in \mathcal{B}(\mathcal{K}^{\otimes 2})$ is unitary and self-adjoint, hence has the form $S=P_+-P_-$ in terms of its eigenprojectors $P_\pm$ for eigenvalues $\pm 1$. Further, we require that $S$ has a parity-invariant diagonal, which is to be understood as
\begin{equation}\label{eq:pinvdiag}
    [S,\mathbb{F}]I_{\otimes 2} = 0.
\end{equation}
This setup yields two important simplifications. 

First, for constant $S$ with parity-invariant diagonal, one easily shows that
\begin{equation}\label{eq:constscatttensor1}
    F(\zeta) := \left( P_+  -i \sh \tfrac{\zeta}{2} P_- \right) I_{\otimes 2}
\end{equation}
satisfies the conditions \ref{f2poles} to \ref{f2parity} from Theorem~\ref{thm:tformgen} with respect to $S$. Thus, $F_2^{\mu\nu}$ as given in Eq.~\eqref{eq:tformgen} is a parity-covariant stress-energy tensor at one-particle level. 
    
Second, for constant $S$, the form factor equations for $F_n$, $n > 2$ simplify significantly; the residue formula connecting $F_n$ with $F_{n-2}$, see Item c in \cite[Sec.~3]{BFK08}, becomes trivial for even $n$. As a consequence, the expression
\begin{equation}\label{eq:tmunu2}
  T^{\mu\nu}(x) := \mathcal{O}_2[F_2^{\mu\nu}(\cdot;x)], 
\end{equation}
reducing the usually infinite expansion \eqref{eq:arakiexpansion} to a single term, is a local operator after time-averaging. In fact, locality may be checked by direct computation from \ref{tmero}--\ref{tper}. Moreover, properties \ref{therm}--\ref{tparity} mean that $T^{\mu\nu}$ is hermitian, is a symmetric covariant two-tensor-valued field with respect to $U_1(x,\lambda)$ (properly extended from Eq.~\eqref{eq:poincarerep} to the full state space), integrates to the total energy-momentum operator $P^\mu = \int ds\, T^{\mu 0}(t,s)$, and is conserved, $\partial_\mu T^{\mu\nu} = 0$. Hence, $T^{\mu\nu}$ is a valid candidate for the stress-energy tensor of the interacting model. While our axioms certainly do not fix $T^{\mu\nu}$ uniquely, $T^{\mu\nu}$ as given in \eqref{eq:tmunu2} constitutes a minimal choice and agrees with the canonical one in models such as the free massive scalar and Majorana field as well as the Ising model \cite{FE98,Daw06,BCF13}. 

For this $T^{\mu\nu}$, we aim to establish a QEI result. 
Our main technique is an estimate for two-particle form factors of a specific factorizing form, which can be stated as follows.
\begin{lem}\label{lem:o2est}
Let $h: \mathbb{S}(0,\pi) \to \mathcal{K}$ be analytic with $L^2$ boundary values at $\mathbb{R}$ and $\mathbb{R}+i\pi$. For
\begin{equation}\label{eq:gfromf}
  f := \operatorname{Symm}_S h\otimes Jh(\bar{\cdot}+i\pi),
\end{equation}
we have in the sense of quadratic forms on $\mathcal{D} \times \mathcal{D}$,
\begin{equation}\label{eq:o2lower}
    \mathcal{O}_2[f] \geq - \frac{1}{2} \|h(\cdot + i\pi)\|_2^2 \mathbbm{1}.
\end{equation}
\end{lem}
\begin{proof}
From the ZF algebra relations in \eqref{eq:algalt1}--\eqref{eq:algalt3}, one verifies that
\begin{equation}\label{eq:comprule}
  \mathcal{O}_1[h] \mathcal{O}_1[h]^\dagger = 2 \mathcal{O}_2[f] + \|h(\cdot + i\pi)\|_2^2 \mathbbm{1}.
\end{equation}
The left-hand side is positive as a quadratic form, implying the result.
\end{proof}
 
Our approach is to decompose $F_2^{00}$ into sums and integrals over terms of the factorizing type \eqref{eq:gfromf} with positive coefficients, then applying the estimate \eqref{eq:o2lower} to each of them.

To that end, we will call a vector  $X \in \mathcal{K}^{\otimes 2}$ \emph{positive} if 
\begin{equation}
 \forall u \in \mathcal{K}:   ( u \otimes Ju, X ) \geq 0.
\end{equation}
This is equivalent to $X$ being a sum of mutually orthogonal vectors of the form $e\otimes Je$ with positive coefficients.\footnote{Vectors of the form $e \otimes Je$ are certainly positive since $(u\otimes Ju, e \otimes Je) = |(u,e)|^2 \geq 0$ and remain positive when summed with positive coefficients. Conversely, given a positive $X$, we note that $\hat{X} \in\mathcal{B}(\mathcal{K})$ is a positive matrix, $(u, \hat{X} u) = (u\otimes Ju,X) \geq 0$, whose eigendecomposition is of the required form.} We also recall the notion of a vector \emph{diagonal in mass}, Eq.~\eqref{eq:diaginmass}. Now we establish our master estimate as follows:

\begin{lem}\label{lem:masterest}
Fix $n \in \{0,1\}$. Suppose that $X\in\mathcal{K}^{\otimes 2}$ is positive, diagonal in mass, and that $SX = (-1)^n X$. Let $h: \mathbb{S}(0,\pi) \to \mathbb{C}$ be analytic with continuous boundary values at $\mathbb{R}$ and $\mathbb{R}+i\pi$ such that $|h(\zeta)| \leq a \exp (b |\Re \zeta|)$ for some $a,b > 0$. Let $g \in \mathcal{D}_\mathbb{R}(\mathbb{R})$. Set
\begin{equation}\label{eq:f2masterdef}
   F_2:=\operatorname{Symm}_S  \left( \bzeta \mapsto h(\zeta_1) \overline{h(\bar\zeta_2+ i \pi)} (\ch \zeta_1 - \ch \zeta_2)^n \widetilde{g^2} (P_0(\bzeta)) X \right).
\end{equation}
Then, in the sense of quadratic forms on $\mathcal{D} \times \mathcal{D}$, it holds that
\begin{equation}\label{eq:masterest}
  \mathcal{O}_2[F_2] \geq - \int_0^\infty  \frac{d\nu}{4\pi}(2\nu)^n \left(I_{\otimes 2}, \Mdiag \left( N_+(\nu,\Mdiag) + N_-(\nu,\Mdiag) \right) X \right)_{\mathcal{K}^{\otimes 2}} \mathbbm{1},
\end{equation}
where the integral is convergent and where
\begin{equation}
   N_\pm(\nu,m) = \| h(\cdot + \tfrac{1\pm 1}{2} i\pi) \tilde g (p_0(\cdot;m) + m\nu) \|_2^2.
\end{equation}
\end{lem}

\begin{proof}
   Since $X$ is diagonal in mass, we have $X = \sum_{m\in\mathfrak{M}} E_m^{\otimes 2} X$. Here, each $E_m^{\otimes 2}X$ is positive, diagonal in mass and, by \ref{strans}, satisfies $S E_m^{\otimes 2} X = E_m^{\otimes 2} S X = (-1)^n E_m^{\otimes 2}X$. As a consequence, we may assume without loss of generality that $X = E_m^{\otimes 2} X$.
   
   Moreover, by positivity of $X$, we may decompose $X=\sum_{\alpha=1}^r c_\alpha\, e_\alpha \otimes Je_\alpha$ with $r \in \mathbb{N}$, $c_\alpha > 0$ and orthonormal vectors $e_\alpha \in \mathcal{K}$, $\alpha=1, \ldots ,r$. Let
   \begin{equation}h^+_{\nu,\alpha}(\zeta) = h(\zeta) \tilde g( p_0(\zeta) - \nu) e_\alpha, \qquad h^-_{\nu,\alpha}(\zeta) = \overline{h^+_{-\nu,\alpha}(\bar\zeta+i\pi)}\end{equation}
 and let $f^\pm_{\nu,\alpha}$ relate to $h^\pm_{\nu,\alpha}$ as in Eq.~\eqref{eq:gfromf}. Further define $f^\pm_{\nu} := \sum_{\alpha=1}^r c_\alpha f^\pm_{\nu,\alpha}$. Since $S X = (-1)^nX$ and $g$ is real-valued, one finds $(-1)^n f^+_{-\nu} = f^-_\nu$ by a straightforward computation.

    Now, in \eqref{eq:f2masterdef} use the convolution formula ($n \in \{0,1\}$, $p_1,p_2 \in\mathbb{C}$),
   \begin{equation}\label{eq:convolthm}
         (p_1-p_2)^n \widetilde{g^2}(p_1+p_2) = \int_{-\infty}^\infty \frac{d\nu}{2\pi} (2\nu)^n \tilde g(p_1-\nu) \overline{\tilde g(-\bar{p}_2-\nu)},
   \end{equation}
   then split the integration region into the positive and negative half\/lines, and obtain
   \begin{equation}\label{eq:f2masterproofeq}
        F_2(\bzeta) = \int_0^\infty \frac{d\nu}{2\pi} (\tfrac{2\nu}{m})^n\left(f^+_\nu(\bzeta) + (-1)^n f^+_{-\nu}(\bzeta) \right) = \int_0^\infty \frac{d\nu}{2\pi} (\tfrac{2\nu}{m})^n \left(f^+_\nu(\bzeta) + f^-_\nu(\bzeta) \right).
   \end{equation}
   Noting that the $h_{\nu,\alpha}^\pm$ are square-integrable at the boundary of $\mathbb{S}[0,\pi]$, we can now apply Lemma~\ref{lem:o2est} to each $f^\pm_{\nu,\alpha}$; then, rescaling $\nu \to m\nu$ in the integral \eqref{eq:f2masterproofeq} yields the estimate \eqref{eq:masterest}.
   Note here that the integration in $\nu$ can be exchanged with taking the expectation value $\braket{\Psi, \mathcal{O}_2[\cdot] \Psi}$, since the integration regions in $\bzeta$ are compact for $\Psi\in\mathcal{D}$, and the series in \eqref{eq:arakiexpansion} is actually a finite sum.
   
   Lastly, we show that the r.h.s of Eq.~\eqref{eq:masterest} is finite. By the Cauchy-Schwarz inequality, the integrand is bounded by a constant times $\nu^n$ times
   \begin{equation}
       N_+(\nu,m)+N_-(\nu,m) = \int d\theta (|h(\theta)|^2 + |h(\theta+i\pi)|^2) |\tilde{g}(m\ch \theta + \nu)|^2.
   \end{equation}
    By assumption, $|h(\theta)| \leq a (\ch \theta)^b$ for some $a,b > 0$, and the resulting integrand $\nu^n (\ch \theta)^{2b} |\tilde{g}(m\ch \theta+m\nu)|^2$ can be shown to be integrable in $(\theta,\nu)$ over $\mathbb{R}\times [0,\infty)$ by substituting $s=\ch \theta+\nu$ \makebox{$(1\leq s < \infty, 0 \leq \nu \leq s-1)$}, and using the rapid decay in $s$ by the corresponding property of $\tilde{g}$. In conclusion, the $\theta$- and $\nu$-integrals converge by Fubini-Tonelli's theorem.
\end{proof}
Now we can formulate:
\begin{thm}[QEI for constant S-functions]\label{thm:qeiconstS}
Consider a constant S-function $S \in \mathcal{B}(\mathcal{K}^{\otimes 2})$ with a parity-invariant diagonal, i.e., $[S,\mathbb{F}]I_{\otimes 2} = 0$ and denote its eigenprojectors with respect to the eigenvalues $\pm 1$ by $P_\pm$. Suppose that $P_\pm I_{\otimes 2}$ are both positive. Then for the energy density $T^{00}(x)$ in Eq.~\eqref{eq:tmunu2} and any $g \in \mathcal{D}_\mathbb{R}(\mathbb{R})$, one has in the sense of quadratic forms on $\mathcal{D} \times \mathcal{D}$:
\begin{equation}\label{eq:constqei}
    T^{00} (g^2) \geq - \left(I_{\otimes 2}, (W_+(M) P_+ + W_-(M) P_-) I_{\otimes 2} \right)_{\mathcal{K}^{\otimes 2}} \mathbbm{1},
\end{equation}
where
\begin{equation}
    W_\pm(m) = \frac{m^3}{4\pi^2} \int_1^\infty ds\, |\tilde{g}(m s)|^2 w_\pm(s) < \infty
\end{equation}
and $w_\pm(s) = s\sqrt{s^2-1} \pm \log(s+\sqrt{s^2-1})$.
\end{thm}

In the scalar case, $\dim \mathcal{K} = 1$, this bound agrees with the previously known bounds for the free massive scalar and Majorana field as well as the Ising model; see Remark~\ref{rem:diagmodels} below.

\begin{proof}
   We use Lemma~\ref{lem:masterest} five times: with $h_1(\zeta) = \ch \zeta$, $h_2(\zeta) = \sh \zeta$, $h_3(\zeta) = 1$ (all with $n=0$ and $X = P_+ I_{\otimes 2}$) and $h_4(\zeta)=\ch\frac{\zeta}{2}$, $h_5(\zeta)=\sh\frac{\zeta}{2}$ (these with $n=1$ and $X =  P_- I_{\otimes 2}$); note that $P_\pm I_{\otimes 2}$ are positive by assumption and diagonal in mass by \ref{strans}.
   Summation of Eq.~\eqref{eq:f2masterdef} for all these five terms and multiplication with $\frac{1}{4\pi}\Mdiagtwo$ yields the expression $\int dt \, g^2(t) F_{2}^{00}(\cdot;(t,0))$ for the energy density in Eq.~\eqref{eq:tmunu2}.  From Lemma~\ref{lem:masterest} we obtain
   \begin{equation}T^{00}(g^2) \geq - \sum_{i=1}^5 \sum_\pm \int_0^\infty \frac{d\nu}{16\pi^2} (2\nu)^{n_i} \big(I_{\otimes 2}, \Mdiagthree N_{\pm,i}(\nu,\Mdiag) P_{s_i} I_{\otimes 2} \big)_{\mathcal{K}^{\otimes 2}} \mathbbm{1}.\end{equation}
Here $s_i := (-1)^{n_i}$. Now we compute
 \begin{equation}
    \begin{aligned}
        & \sum_{i=1}^5 \sum_\pm \int_0^\infty \frac{d\nu}{16\pi^2} \, (2\nu)^{n_i} \Mdiagthree \lVert h_{i}(\cdot+\tfrac{1\pm 1}{2} i\pi) \tilde g(P_0(\theta)+\Mdiag\nu)\rVert_2^2 P_{s_i} \\
        &\quad = \frac{\Mdiagthree}{8\pi^2} \int_0^\infty d\nu \int_{-\infty}^\infty d\theta \, \lvert\tilde{g}(P_0(\theta)+\Mdiag \nu)\rvert^2 \left( (1 + \ch^2 \theta + \sh^2 \theta)P_+ + 2\nu (\ch^2 \tfrac{\theta}{2} +\sh^2\tfrac{\theta}{2})P_-\right) \\
        &\quad = \frac{\Mdiagthree}{4\pi^2} \int_0^\infty d\nu \int_{-\infty}^\infty d\theta \, \lvert\tilde{g}(P_0(\theta)+\Mdiag \nu)\rvert^2 \left( \ch^2 \theta P_+ + \nu \ch \theta P_-\right) \\
        &\quad = \frac{\Mdiagthree}{4\pi^2} \int_1^\infty ds |\tilde{g}(\Mdiag s)|^2 (w_+(s) P_+ + w_-(s)P_-) \\
        &\quad = W_+(M) P_+ + W_-(M) P_-,
    \end{aligned}
 \end{equation}
where we have substituted $s=\ch \theta+\nu$ ($1 \leq s < \infty$, $0 \leq \nu \leq s-1$), then solving explicitly the integral in $\nu$.
\end{proof}

\begin{rem}\label{rem:diagmodels}
    The conditions of Theorem~\ref{thm:qeiconstS} are at least fulfilled in (constant) diagonal models, i.e., for S-functions of the form
    $S = \sum_{\alpha\beta} c_{\alpha\beta} | e_\alpha \otimes e_\beta )( e_\beta \otimes e_\alpha |$ for some choice of an orthonormal basis $\{e_\alpha\}$ and coefficients $c_{\alpha\beta}$, where we suppose $Je_\alpha = e_{\bar\alpha}$ as indicated in Remark~\ref{rem:scattinbasis}. The S-function has to satisfy $S = S^\dagger = S^{-1} = \mathbb{F} J^{\otimes 2} S J^{\otimes 2} \mathbb{F}$ which at the level of coefficients becomes $|c_{\alpha\beta}|=1$, $c_{\alpha\beta} = c_{\beta\alpha}^{-1}$ and $c_{\alpha\beta} = c_{\bar\alpha\bar\beta}$. In particular, one has $c_{\alpha\bar{\alpha}} = c_{\bar\alpha \alpha} \in \{\pm 1\}$. Together with $P_\pm = \tfrac{1}{2}(1\pm S)$ this implies $P_\pm I_{\otimes 2} = \sum_{\alpha: c_{\alpha\bar{\alpha}}=\pm 1 }| e_\alpha \otimes Je_\alpha)$ which is clearly positive. Also, $[S,\mathbb{F}]I_{\otimes 2} = 0$ by a straightforward computation using $c_{\alpha\bar\alpha} = c_{\bar\alpha\alpha}$ and $\mathbb{F}I_{\otimes 2} = I_{\otimes 2}$.
    
    Thus, the QEI applies to all such models. This does not only include the known QEI results for the free Bose field \cite{FE98}, the free Fermi field \cite{Daw06}, the Ising model \cite{BCF13}, and combinations of those, but also the symplectic model, a fermionic variant of the Ising model (see, e.g., \cite{Las94} or \cite{BC21}).

    It also applies to the Federbush model (and generalizations of it as in \cite{Tan14}): Although the Federbush model's S-function is not parity invariant, it has a parity invariant diagonal and Eq.~\eqref{eq:constscatttensor1} yields a valid (parity covariant) candidate for the stress-energy tensor, i.e., it satisfies all the properties \ref{f2poles} to \ref{f2parity}. The candidate is in agreement with \cite[Sec.~4.2.3]{CF01}. For further details on the Federbush model, see Section~\ref{sec:fb}.
\end{rem}

\begin{rem}
    The QEI result is independent of the statistics of the particles; it depends only on the mass spectrum and the S-function. The aspect of particle statistics comes into play when computing the scattering function from the S-function (see Remark~\ref{rem:scattfct}); it also enters the form factor equations for local operators (see, e.g., \cite[Sec.~6]{BC21}). However, in the equations for $F_2$ relevant for our analysis, the ``statistics factors'' occur only in even powers, so that our assumptions on the stress-energy tensor---specifically, properties \ref{tssym} and \ref{tper} in Def.~\ref{def:set1}---are appropriate in both bosonic and fermionic cases.
\end{rem}

\begin{rem}
    In the short-distance scaling limit, corresponding to $m \to 0$ with $M=m \mathbbm{1}$ at fixed $g$ \cite{BLM11}, the QEI bound simplifies and becomes proportional to the number of degrees of freedom in the model,
    \begin{equation}
        \mathrm{r.h.s.}\text{ of }\eqref{eq:constqei} \quad \to \quad \dim \mathcal{K} \cdot \frac{1}{4\pi} \int |g'(s)|^2 ds.
    \end{equation}
    Comparing with optimal bounds in the known cases, the free massless scalar and Majorana field, this bound is not optimal: It is larger by a constant factor $\tfrac{3}{2}$ (scalar), resp., $3$ (Majorana), as has been noted before in \cite{BCF13}.
\end{rem}

\newpage

\section{QEI at one-particle level for general integrable models}\label{sec:onepQEI}

This section aims to give necessary and sufficient conditions for QEIs at one-particle level in general integrable models, including models with several particle species and bound states. The conditions are expressed in Theorem~\ref{thm:qeimain} in cases \ref{item:noqei} and \ref{item:qei}, respectively. 

Given a stress-energy tensor $F^{\mu\nu}_2$ at one-particle level, including diagonality in mass, the expectation values of the averaged energy density are, combining Eq.~\eqref{eq:expvalues} with Corollary~\ref{cor:diagmass}, given by
\begin{equation}\label{eq:52}
   \braket{\varphi, T^{00}(g^2) \varphi} 
   = \int d\theta\, d\eta
    \ch^2 \frac{\theta+\eta}{2} \Big( \varphi(\theta),  \frac{M^2}{2\pi} \widetilde{g^2}(p_0(\theta;M)-p_0(\eta;M)) \hat{F}(\eta-\theta+i\pi)  \varphi(\eta) \Big)
\end{equation}
for $\varphi \in \mathcal{D}\cap \mathcal{H}_1$. We ask whether this quadratic form is bounded below. In fact, this can be characterized in terms of the asymptotic behaviour of $\hat{F}$:
\begin{thm}\label{thm:qeimain}
    Let $F_2^{\mu\nu}$ be a parity-covariant stress-energy tensor at one-particle level which is diagonal in mass and $\hat{F}$ be given according to Corollary~\ref{cor:diagmass}. Then: 
    \begin{enumerate}[label=(\alph*)]
     \item \label{item:noqei} Suppose there exists $u \in \mathcal{K}$ with $\lVert u \rVert_\mathcal{K}=1$, and $c > \tfrac{1}{4}$ such that
    \begin{equation}
        \exists r>0 \, \forall |\theta|\geq r: \quad \lvert (u, \hat{F}(\theta+i\pi) u)\rvert \geq c \exp |\theta|. \label{eq:fposbound}
    \end{equation}
    Then for all $g \in \mathcal{S}_\mathbb{R}(\mathbb{R})$, $g\neq 0$ there exists a sequence $(\varphi_j)_j$ in $\mathcal{D}(\mathbb{R},\mathcal{K}),\, \lVert \varphi_j  \rVert_2 = 1$, such that
    \begin{equation}\label{eq:noQEI}
        \braket{\varphi_j, T^{00}(g^2) \varphi_j} \xrightarrow{j\to \infty} -\infty.
    \end{equation}
    \item \label{item:qei} Suppose there exists $0 < c < \tfrac{1}{4}$ such that
    \begin{equation}\label{eq:fopbound}
        \exists \epsilon,r>0 \, \forall |\Re \zeta| \geq r, |\Im \zeta|\leq \epsilon: \quad \lVert \hat{F}(\zeta+i\pi) \rVert_{\mathcal{B}(\mathcal{K})} \leq c \exp |\Re \zeta|.
    \end{equation}
    Then for all $g\in \mathcal{S}_\mathbb{R}(\mathbb{R})$ there exists $c_g > 0$ such that for all $\varphi \in \mathcal{D}(\mathbb{R},\mathcal{K})$,
    \begin{equation}\label{eq:QEI}
        \braket{\varphi, T^{00}(g^2) \varphi} \geq - c_g \lVert \varphi \rVert_2^2.
    \end{equation}
    \end{enumerate}
\end{thm}

\noindent The two cases are mutually exclusive. While case (b) establishes a QEI at one-particle level \eqref{eq:QEI}, case (a) implies that no such QEI can hold. Before we proceed to the proof, let us comment on the scope of the theorem.

\begin{rem}\label{rem:noparitycovariance}
 We require parity-covariance of $F_2^{\mu\nu}$. In absence of this property, at least the parity-covariant part $F^{\mu\nu}_{2,P}$ of $F_2^{\mu\nu}$, 
 which is given by replacing $F$ with $F_P := \frac{1}{2}(1+\mathbb{F})F$, has all features of a parity-covariant stress-energy tensor at one-particle level except possibly for S-symmetry \ref{tssym}, which requires the extra assumption $[S,\mathbb{F}]F=0$. In any case, since S-symmetry will not be used in the proof, Theorem~\ref{thm:qeimain} still applies to $F^{\mu\nu}_{2,P}$. Now, if \eqref{eq:fposbound} holds for $F$ with $u$ satisfying $Ju = \eta u$ with $\eta \in \mathbb{C}$ and $|\eta|=1$, it holds for $F_P$ due to $(u,\widehat{F}(\theta) u) = (Ju, \widehat{F}(\theta) Ju) = (u,\widehat{\mathbb{F}F}(\theta) u)$. As a consequence, case (a) applies and no QEI can hold for $F_2^{\mu\nu}$. On the other hand, if \eqref{eq:fopbound} is fulfilled for $F$ (hence for $F_P$), then a one-particle QEI for $F_2^{\mu\nu}$ of the form \eqref{eq:QEI} holds at least in parity-invariant one-particle states.
\end{rem}

\begin{rem}\label{rem:superpos01}
    While Theorem~\ref{thm:qeimain}\ref{item:qei} establishes a QEI only at one-particle level, the result usually extends
    to expectation values in vectors $\Psi = c \, \Omega + \Psi_1$, $c\in\mathbb{C}, \Psi_1 \in \mathcal{H}_1$. 
    Namely,
    \begin{equation}
          \braket{\Psi, T^{00}(g^2) \Psi} =  \braket{\Psi_1, T^{00}(g^2) \Psi_1} + 2 \Re \,c \int (\Psi_1(\theta), \widetilde{g^2}(p_0(\theta;M))  F_1(\theta)) d\theta ,
    \end{equation}
    where $F_1=F_1^{[T^{00}(0)]}$ is the one-particle form factor of the energy density. This $F_1$ may be nonzero. However, due to Remark~\ref{rem:f1}, it is of the form $F_1(\zeta;0) = F_1(0) \sh^2\zeta$; thus, the rapid decay of $\widetilde{g^2}$ and the Cauchy-Schwarz inequality imply that the additional summand is bounded in $\lVert \Psi_1 \rVert_2$, hence in $\|\Psi\|^2$. 
\end{rem}

The rest of this section is devoted to the proof of Theorem~\ref{thm:qeimain}, which we develop separately for the two parts \ref{item:noqei} and \ref{item:qei}. We first note that from Theorem~\ref{thm:tformgen}, the operators $\hat{F}(\zeta)$ fulfil
\begin{align}
    \hat{F}(\zeta+i\pi) & = \hat{F}(-\zeta+i\pi), \label{eq:hatfsym}\\
    \hat{F}(\zeta+i\pi) & = \hat{F}(\bar\zeta+i\pi)^\dagger, \label{eq:hatfherm}\\
    \hat{F}(i\pi) &= \mathbbm{1}_\mathcal{K}.\label{eq:hatfnorm}
\end{align}
In more detail, these equations are implied by $S$-periodicity and parity-invariance for \eqref{eq:hatfsym}, by $S$-periodicity and CPT-invariance for \eqref{eq:hatfherm}, and by normalization for \eqref{eq:hatfnorm}.

Now the strategy for part \ref{item:noqei} closely follows \cite[Proposition~4.2]{BC16}, but with appropriate generalizations for matrix-valued rather than complex-valued $\hat{F}$.

\begin{proof}[Proof of Theorem~\ref{thm:qeimain}\ref{item:noqei}]
    Fix a smooth, even, real-valued function $\chi$ with support in $[-1,1]$. Then for $\rho > 0$ define $\chi_\rho(\theta) := \rho^{-1/2} \lVert \chi \rVert_2^{-1} \chi(\rho^{-1} \theta)$, so that $\chi_\rho$ has support in $[-\rho,\rho]$ and is normalized with respect to $\lVert \cdot \rVert_2$. Define $\varphi_j(\theta) := \tfrac{1}{\sqrt{2}} (\chi_{\rho_j}(\theta-j)+ s\, \chi_{\rho_j}(\theta+j)) M^{-1}u$, where $s \in \{\pm 1\}$ and $(\rho_j)_j$ is a null sequence with $0 < \rho_j < 1$; both will be specified later. The $\varphi_j$, thus defined, have norm of at most $m_-^{-1}$, where $m_- := \min \mathfrak{M}$, and \eqref{eq:52} yields
    \begin{equation}
       \braket{\varphi_j, T^{00}(g^2)\varphi_j} = \frac{1}{4\pi} \big(u, (H_{\chi,j,+} + s H_{\chi,j,-}) u \big)     
    \end{equation}
    with $H_{\chi,j,\pm} := \int d\theta d\eta \,\widetilde{g^2}(M k_j(\theta,\eta)) H_{j,\pm}(\theta,\eta)\chi_{\rho_j}(\theta) \chi_{\rho_j}(\eta)$
    and
    \begin{align*}
        H_{j,+}(\theta,\eta) & = \ch^2(j+\tfrac{\theta+\eta}{2})\hat{F}(\theta-\eta+i\pi), \\
        H_{j,-}(\theta,\eta) & = \ch^2 \tfrac{\theta-\eta}{2} \hat{F}(2j+\theta+\eta+i\pi),\\
        k_j(\theta,\eta) & = 2 \sh (j+\tfrac{\theta+\eta}{2}) \sh \tfrac{\theta-\eta}{2}.
    \end{align*}
    We used here \eqref{eq:hatfsym} and that $\chi$ is an even function. For large $j$ and for $\theta,\eta \in [-\rho_j,\rho_j]$, we establish the estimates
    \begin{align}
        (u,H_{j,+}(\theta,\eta)u) & \leq \lVert H_{j,+}(\theta,\eta) \rVert_{\mathcal{B}(\mathcal{K})} \leq (\tfrac{1}{2}+2c)\left(1+\tfrac{1}{4}e^{2j}e^{2\rho_j}\right), \                                                                                                    \label{eq:hj11} \\
        s (u,H_{j,-}(\theta,\eta)u) & \leq -c e^{2j}e^{-2\rho_j}, \label{eq:hj12}\\
        |k_j(\theta,\eta)| & \leq 12 e^{j} \rho_j. \label{eq:kj}
    \end{align}
    Namely for \eqref{eq:hj11}, due to \eqref{eq:hatfnorm} and continuity of $\hat F$ restricted to $\mathbb{R}$, we have $\lVert F(\theta+i\pi)\rVert_{\mathcal{B}(\mathcal{K})} \leq 2c+\frac{1}{2}>1$ for $\theta \in [-2\rho_j,2\rho_j]$ and large $j$. Also, $\ch^2 x \leq 1 + \tfrac{1}{4}e^{2x}$. For \eqref{eq:hj12} one uses $\ch^2 x \geq 1$ along with the estimate $- s (u,\hat{F}(\theta+i\pi)u) \geq c \exp |\theta|$ for all $|\theta| \geq r$, with suitable choice of $s\in \{\pm 1\}$. The latter statement is implied by hypothesis \eqref{eq:fposbound} since $(u,\hat{F}(\theta+i\pi) u)$ is real-valued (due to \ref{eq:hatfherm}) and continuous. For \eqref{eq:kj}, see \cite[Eq.~(4.17)]{BC16}. 
    
    Now choose $\delta > 0$ so small that $\widetilde{g^2}(m_+ p) \geq \frac{1}{2} \widetilde{g^2}(0)>0$ for $|p|\leq \delta$, where $m_+ := \max \mathfrak{M}$. Choosing specifically the sequence $\rho_j = \frac{\delta}{12} e^{-j}$, we can combine these above estimates in the integrands of $H_{\chi,j,\pm}$ to give, cf.~\cite[Proof of Proposition~4.2]{BC16},
    \begin{equation}
        \big(u, (H_{\chi,j,+}+ s H_{\chi,j,-}) u\big) \leq \frac{\delta}{24} \widetilde{g^2}(0) (ce^{-j} - c' e^j) \big(\rho_j^{-1/2} \lVert\chi_{\rho_j}\rVert_1\big)^2 \xrightarrow{j\to \infty} -\infty
    \end{equation}
    with some $c'>0$, noting that $\rho_j^{-1/2} \lVert \chi_{\rho_j} \rVert_1$ is independent of $j$.
\end{proof}

For part \ref{item:qei}, we follow \cite[Theorem~5.1]{BC16}, but again need to take the operator properties of $\hat F$ into account.

\begin{proof}[Proof of Theorem~\ref{thm:qeimain}\ref{item:qei}]
For fixed $\varphi \in \mathcal{D}(\mathbb{R},\mathcal{K})$ and $g\in\mathcal{S}_\mathbb{R}(\mathbb{R})$, we introduce
$X_\varphi:=\braket{\varphi, T^{00}(g^2) \varphi}$. Our aim is to decompose $X_\varphi = Y_\varphi + (X_\varphi - Y_\varphi)$ with $Y_\varphi \geq 0$ and $\lvert X_\varphi - Y_\varphi \rvert \leq c_g \lVert \varphi \rVert^2_2$ in order to conclude $X_\varphi \geq -c_g \lVert \varphi \rVert_2^2$. Since $[M,\hat{F}(\zeta)]=0$ from diagonality in mass, we have $X_\varphi = \sum_{m\in\mathfrak{M}} X_{E_m\varphi}$ and can treat each $E_m \varphi$, $m\in\mathfrak{M}$, separately. Therefore in the following, we assume $M = m \mathbbm{1}_\mathcal{K}$ without loss of generality.

We now express $X_\varphi$ as in \eqref{eq:52} and rewrite the integral as
\begin{equation}
    X_\varphi = \frac{m^2}{2\pi} \int_0^\infty \int_0^\infty d\theta d\eta \, \widetilde{g^2}(p_0(\theta)-p_0(\eta)) \left(\underline{\varphi}(\theta)^t, \underline{\underline{X}}(\theta,\eta) \underline{\varphi}(\eta)\right),
\end{equation}
where $\underline\varphi(\theta)=(\varphi(\theta),\varphi(-\theta))^t$ and 
\begin{equation*}
 \underline{\underline{X}}(\theta,\eta) = \left( \begin{matrix} \ch^2 \tfrac{\theta+\eta}{2}\hat{F}(-\theta+\eta+i\pi) & \ch^2 \tfrac{\theta-\eta}{2}\hat{F}(-\theta-\eta+i\pi) \\ \ch^2 \tfrac{-\theta+\eta}{2}\hat{F}(\theta+\eta+i\pi) & \ch^2 \tfrac{\theta+\eta}{2}\hat{F}(\theta-\eta+i\pi) \end{matrix} \right).
\end{equation*}
Using \eqref{eq:hatfsym}, we find $\underline{\underline{X}} =  \left( \begin{smallmatrix} A & B \\ B & A\end{smallmatrix} \right)$ with 
$$A(\theta,\eta) = \ch^2 \tfrac{\theta+\eta}{2} \hat{F}(\theta-\eta+i\pi), \quad B(\theta,\eta) = \ch^2 \tfrac{\theta-\eta}{2} \hat{F}(\theta+\eta+i\pi).$$
Defining $H_\pm = A \pm B$ and $\varphi_\pm(\theta) = \varphi(\theta) \pm \varphi(-\theta)$ we obtain further that
\begin{equation}(\underline\varphi(\theta)^t ,\underline{\underline{X}}(\theta,\eta) \underline\varphi(\eta)) = \sum_\pm (\varphi_\pm(\theta), H_\pm(\theta,\eta) \varphi_\pm(\eta)).\end{equation}
Let us define
\begin{equation}K_\pm(\theta) := \sqrt{ |H_\pm(\theta,\theta)|} \in \mathcal{B}(\mathcal{K}),\end{equation}
where for $O \in \mathcal{B}(\mathcal{K})$, $|O|$ denotes the operator modulus of $O$ and $\sqrt{|O|}$ its (positive) operator square root.
Now, analogous to $X_\varphi$, introduce $Y_\varphi$ (replacing $H_\pm(\theta,\eta)$ with $K_\pm(\theta) K_\pm(\eta)$),
\begin{equation}Y_\varphi := \frac{m^2}{2\pi} \sum_{\pm} \int_0^\infty \int_0^\infty d\theta d\eta \, \widetilde{g^2}(p_0(\theta)-p_0(\eta)) \left(\varphi_\pm(\theta), K_\pm(\theta)K_\pm(\eta)\varphi_\pm(\eta)\right) .\end{equation}
Using the convolution formula \eqref{eq:convolthm} with $n=0$, $p_1=p_0(\theta)$, $p_2 = p_0(\eta)$, noting that for real arguments it also holds for $g\in \mathcal{S}_\mathbb{R}(\mathbb{R})$, one finds that
\begin{equation}Y_\varphi = \frac{m^2}{2\pi}\sum_\pm \int \frac{d\nu}{2\pi} \left\lVert \int d\eta \, \psi_\pm (\eta,\nu)\right\rVert_\mathcal{K}^2 \geq 0, \;\; \text{where} \; \psi_\pm(\eta,\nu):= \widetilde{g}(p_0(\eta)+\nu) K_\pm(\eta)\varphi_\pm(\eta).\end{equation}
It remains to show that $|X_\varphi - Y_\varphi| \leq c_g \lVert \varphi \rVert_2^2$ for some $c_g \geq 0$. For this it suffices to prove that
\begin{equation}
    c_g := \sum_\pm \int_0^\infty d\theta \int_0^\infty d\eta |\widetilde{g^2}(p_0(\theta)-p_0(\eta))|^2 \lVert H_\pm(\theta,\eta)-K_\pm(\theta)K_\pm(\eta) \rVert_{\mathcal{B}(\mathcal{K})}^2 \label{eq:finitec}
\end{equation}
is finite. 

To that end, let us introduce $L_\pm(\rho,\tau) := H_\pm(\rho+\tfrac{\tau}{2},\rho-\tfrac{\tau}{2}) \pm K_\pm(\rho+\tfrac{\tau}{2})K_\pm(\rho-\tfrac{\tau}{2})$, where $\rho = \tfrac{\theta+\eta}{2}$, $\tau = \theta -\eta$, and $|\partial(\rho,\tau) / \partial(\theta,\eta)| = 1$.
In these coordinates, the integration region in \eqref{eq:finitec} is given by $\rho > 0$, $|\tau| < 2\rho$. 
Let $\rho_0 \geq 1$ and $\theta_0 > 0$ be some constants.
The region $\rho \leq \rho_0$ is compact; thus, the integral over this region is finite.
The region $\rho > \rho_0, |\tau| > 1$ also gives a finite contribution: Because of
\begin{equation}|p_0(\theta)-p_0(\eta)|= 2m \sh \tfrac{|\tau |}{2} \sh \rho  \geq 2m (1-e^{-2\rho_0}) \sh\tfrac{1}{2} \ch \rho \end{equation}
in this region, $|\widetilde{g^2}(p_0(\theta)-p_0(\eta))|^2$ decays faster than any power of $\ch \rho$, while $\lVert L_\pm(\rho,\tau)\rVert_{\mathcal{B}(\mathcal{K})}^2$ cannot grow faster than a finite power of $\ch \rho$ due to our hypothesis \eqref{eq:fopbound}.
The remaining region is given by $\rho \geq \rho_0$ and $|\tau| \leq 1$. By \eqref{eq:fopbound}, there exists $0 < c < \tfrac{1}{4}$ and $r>0$ such that
\begin{equation}
    \forall \theta\geq r: \, ||\hat{F}(2\theta+i\pi)||_{\mathcal{B}(\mathcal{K})} \leq c \exp 2|\theta| \leq 4c \ch^2 \theta.
\end{equation}
This implies, also using self-adjointness of $\hat{F}$ (see \eqref{eq:hatfherm}), that for all $\theta \geq r$:
\begin{equation}\label{eq:hestimate}
    H_\pm(\theta,\theta) = \ch^2 \theta \,\hat{F}(i\pi) \pm \hat{F}(2\theta+i\pi) \geq \ch^2\theta \, \mathbbm{1}_{\mathcal{K}} - |\hat{F}(2\theta+i\pi)| \geq (1-4c) \ch^2 \theta \,\mathbbm{1}_\mathcal{K}.
\end{equation}
Since $c < \frac{1}{4}$, these $H_\pm(\theta,\theta)$ are positive operators with a uniform spectral gap at $0$. As a consequence, together with $H_\pm(\theta,\theta)$, also the maps $\theta \mapsto K_\pm(\theta) = \sqrt{H_\pm(\theta,\theta)}$  are analytic near $[r,\infty)$; see \cite[Ch.~VII, \S{}5.3]{Kat95}. Correspondingly, $L_\pm(\rho,\tau)$ is real-analytic in the region where $\rho \geq \tfrac{|\tau|}{2}+r$. This contains the region $\{(\rho,\tau):\rho \geq \rho_0, |\tau|\leq 1\}$ if we choose $\rho_0 \geq \tfrac{1}{2} + r$.

Now in this region, it can be shown that there exists $a>0$ such that for any normalized $u\in\mathcal{K}$,
\begin{equation}
    \big\lvert \big(u, L_\pm(\rho,\tau)u \big)\big\rvert \leq \tfrac{1}{2} \tau^2 \underset{|\xi|\leq 1}{\sup} \big\lvert \big( u, \tfrac{\partial^2}{\partial \xi^2} L_\pm(\rho,\xi)u \big) \big\rvert \leq \tfrac{1}{2} a \tau^2 \ch \rho.\label{eq:lgrowth}
\end{equation}
This estimate is based on the fact that $L_\pm(\rho,\tau) = L_\pm(\rho,-\tau)$, and  $L_\pm(\rho,0) = 0$ (which also uses positivity of $H_\pm$).
The first inequality in \eqref{eq:lgrowth} then follows from Taylor's theorem; the second is an estimate of the derivative by Cauchy's formula, using analyticity of $\hat{F}(\cdot +i\pi)$ in a strip around $\mathbb{R}$, and repeatedly applying the estimate \eqref{eq:fopbound}, cf.~\cite[Proof of Lemma~5.3]{BC16}. Since \eqref{eq:fopbound} is an estimate in operator norm, and the other parts of the argument are $u$-independent, one finds  $\lVert \tfrac{\partial^2}{\partial \xi^2}  L_\pm(\rho,\tau)\rVert_{\mathcal{B}(\mathcal{K})} \leq a \ch \rho$ with a constant $a$.

Finiteness of the integral \eqref{eq:finitec} now follows from the estimate \eqref{eq:lgrowth} together with $|\widetilde{g^2}(p_0(\theta)-p_0(\eta))| \leq c'(\tau^4 \ch^4 \rho +1)^{-1}$ for some $c'>0$; cf.~\cite[Proof of Lemma~5.4]{BC16}.
\end{proof}

\newpage

\section{The connection between the S-function\\ and the minimal solution} \label{sec:minimal}

For the purpose of analysing particular examples, it is helpful to introduce the \emph{minimal solution} of a model, a well-known concept in the form factor programme \cite{KW78} which plays an essential role in the description and classification of the observables of the model. We will here give a brief summary of necessary facts for the examples in Section~\ref{sec:examples} and a recipe for obtaining QEIs for other models. For technical details and full proofs, we refer to Appendix~\ref{app:fmin}.

Given an S-function, in generic cases including diagonal models and all our examples we can perform an eigenvalue decomposition into meromorphic complex-valued functions $S_{i}$ and meromorphic projection-valued functions $P_i$ such that
\begin{equation}
S(\zeta) = \sum_{i=1}^k S_{i}(\zeta) P_i(\zeta)
\end{equation}
(see Proposition~\ref{prop:sedecomp}). For each eigenfunction $S\equiv S_{i}$ (omitting the index $i$ for the moment), the \emph{minimal solution} is a meromorphic function $F_{\mathrm{min}}:\mathbb{C}\to\mathbb{C}$ which is the most regular solution of the form factor equations at one-particle level (or Watson's equations),
\begin{equation}\label{eq:watson}
 F_{\mathrm{min}}(\zeta) = S(\zeta) F_{\mathrm{min}}(-\zeta), \quad F_{\mathrm{min}}(\zeta+i\pi) = F_{\mathrm{min}}(-\zeta+i\pi),
\end{equation}
subject to the normalization condition $F_{\mathrm{min}}(i\pi) = 1$ (see Appendix~\ref{sec:fminunique}). A general solution to \eqref{eq:watson} is then of the form
\begin{equation}\label{eq:1partsolution}
    F_q(\zeta) = q(\ch \zeta) F_{\mathrm{min}}(\zeta),
\end{equation}
where $q$ is a rational function which is fixed by the pole- and zero-structure of $F_q$, and $q(-1) = 1$ if $F_q(i\pi) = 1$ (Lemma~\ref{lem:fgenform}).

Uniqueness of $F_\mathrm{min}$ follows under mild growth conditions (Lemma~\ref{lem:minsol}). Existence can be proved for a large class of (eigenvalues of) S-functions by employing a well-known integral representation. For this class, the function
\begin{equation}\label{eq:fdef}
    f[S]:\mathbb{R}\to \mathbb{R}, \quad t\mapsto f[S](t) := -\tfrac{1}{\pi} \int_0^\infty S'(\theta)S(\theta)^{-1} \cos (\pi^{-1} \theta t) d\theta
\end{equation}
is well-defined and referred to as the \emph{characteristic function} of $S$. In the case $S(0)=1$, the minimal solution is then obtained from $f=f[S]$ as the meromorphic continuation of 
\begin{equation}
    F_f:\mathbb{R}\to \mathbb{C}, \quad \theta \mapsto F_f(\theta) := \exp \left(2\int_0^\infty f(t) \sin^2 \frac{(i\pi-\theta) t}{2\pi} \, \frac{dt}{t\sh t}\right) . \label{eq:fint}
\end{equation}
For $S(0)=-1$, an additional factor needs to be included (see Theorem~\ref{thm:propsofintreps}). 

For our analysis of QEIs, it will be crucial to control the large-rapidity behaviour of $F_\mathrm{min}$ using properties of the characteristic function $f[S]$. This is in fact possible as follows (Proposition~\ref{prop:minasympt}): For a continuous function $f:[0,\infty)\to \mathbb{R}$, which is exponentially decaying at large arguments and second-order differentiable on some interval $[0,\delta], \delta > 0$, and where $f_0:=f(0)$, $f_1:=f'(0)$, the growth of $F_f(\zeta)$ is bounded at large $|\Re \zeta|$ as in
\begin{equation}\label{eq:fminasgrowth}
    \exists 0 < c \leq c',\,r>0: \, \forall |\Re\zeta|\geq r, \Im \zeta \in [0,2\pi]: \quad c \leq \frac{|F_f(\zeta)|}{|\Re \zeta|^{f_1} \exp |\Re \zeta|^{f_0/2}} \leq c'.
\end{equation}  
With this said, we have a recipe for a large class of models to determine whether a one-particle QEI in the sense of Theorem~\ref{thm:qeimain} holds, or no such QEI can hold: According to Theorem~\ref{thm:tformgen} and Corollary~\ref{cor:diagmass}, we know that
\begin{equation}F_2^{\mu\nu}(\bzeta;0) = G_{\mathrm{free}}^{\mu\nu}(\tfrac{\zeta+\zeta'}{2}) F(\zeta'-\zeta).\end{equation}
Then $F$ can be decomposed into the eigenbasis with respect to $S$, namely $F(\zeta) := \sum_{i=1}^k F_i(\zeta)$, where $F_i(\zeta) := P_i(\zeta) F(\zeta)$. Let us restrict to parity-invariant $F$ and constant eigenprojectors $P_i$, i.e., having $F = \mathbb{F} F$ and $P_i=const$. Then (in some orthonormal basis) the components of each $F_i$ will satisfy Watson's equations and take the form as in Eq.~\eqref{eq:1partsolution}. Therefore, each $F_i$ will be of the form $F_i(\zeta) = Q_i(\ch \zeta) F_{i,\mathrm{min}}(\zeta)$, where $Q_i$ is a rational function that takes values in $\mathcal{K}^{\otimes 2}$ and $F_{i,\mathrm{min}}$ is the minimal solution with respect to $S_i$. In case of symmetries, the choice of $Q_i$ is further restricted by $\mathcal{G}$-invariance. The asymptotic growth of the $F_i$ will be bounded by the growth of the $Q_i$ and the bound \eqref{eq:fminasgrowth} for the $F_{i,\mathrm{min}}$. In summary, depending on the growth of the $Q_i$ and the $F_{i,\mathrm{min}}$, we can determine the asymptotic growth of $F$ and thus decide whether a one-particle QEI holds or not.

\newpage
\section{QEIs in examples}\label{sec:examples}
We now discuss some examples of integrable models which illustrate essential features of the abstract results developed in Sections \ref{sec:constantS} and \ref{sec:onepQEI}. These include a model with bound states (Bullough-Dodd model, Sec.~\ref{sec:bd}), an interacting model with a constant scattering function (Federbush model, Sec.~\ref{sec:fb}), and a model with several particle species ($O(n)$-nonlinear sigma model, Sec.~\ref{sec:onsigma}).\\ 

As a first step, we review in our context the known results for models of one scalar particle type and without bound states \cite{BC16}. That is, we consider $\mathcal{K}=\mathbb{C}$,  $J$ the complex conjugation, $\mathfrak{M}=\{m\}$ for the one-particle space, and $\mathfrak{P}=\emptyset$ for the stress-energy tensor, with a scattering function of the form 
\begin{equation}\label{eq:prodshG}
    S(\zeta) = \epsilon \prod_{k=1}^n S(\zeta;b_k), \quad S(\zeta;b) := \frac{\sh \zeta -i \sin \pi b}{\sh \zeta + i \sin \pi b},
\end{equation}
where $\epsilon = \pm 1, n\in \mathbb{N}_0,$ and $(b_k)_{k\in\{1, \ldots ,n\}} \subset i\mathbb{R}+(0,1)$ is a finite sequence in which $b_k$ and $\overline{b_k}$ appear the same number of times.

The minimal solution with respect to $\zeta \mapsto S(\zeta;b)$ is known---see, e.g., \cite[Eq.~\lParent 2.5\rParent ]{BC16} or \cite[Eq.~\lParent 4.13\rParent ]{FMS93}---and in our context given by
\begin{equation}\label{eq:shgfmin}
    F_{b,\mathrm{min}}(\zeta) = (-i\sh \tfrac{\zeta}{2}) F_{f(\cdot;b)}(\zeta), \quad f(t;b) := \frac{4\sh\tfrac{bt}{2} \sh\tfrac{(1-b)t}{2}\sh\tfrac{t}{2}-\sh t}{\sh t}.
\end{equation}
Since $f(t;b) = -1 + \mathcal{O}(t^2)$ for $t\to 0$, it follows that $F_{b,\mathrm{min}}$ is uniformly bounded above and below on $\mathbb{S}[0,2\pi]$ by Proposition~\ref{prop:minasympt}. More quantitatively, $F_{b,\mathrm{min}}(\zeta+i\pi)$ converges uniformly to 
\begin{equation}\label{eq:shgasympt}
    F_{b,\mathrm{min}}^{\infty} := \lim_{\theta\to \pm \infty} F_{b,\mathrm{min}}(\theta+i\pi) = \exp \int_{\mathbb{R}} (t\sh t)^{-1} (1+f(t;b)) dt < \infty
\end{equation}
for $|\Re \zeta| \to \infty$ and $|\Im \zeta| \leq \delta$ for any $0<\delta< \pi$. 

This can be derived in the following way: Since $g(t):=(t\sh t)^{-1} (1+f(t;b))$ is exponentially decaying and regular (in particular at $t=0$), it is integrable and $F_{b,\mathrm{min}}^\infty$ is finite. As $\log \ch \tfrac{\zeta}{2} = 2\int_0^\infty (t\sh t)^{-1} \sin^2 \tfrac{\zeta t}{2\pi} dt$ for $|\Im \zeta| < \pi$ one may write $\log F_{b,\mathrm{min}}(\zeta+i\pi) = 2 \int_{\mathbb{R}} (t\sh t)^{-1} (1+f(t;b)) \sin^2 \tfrac{\zeta t}{2\pi} dt$. In the limit $|\Re \zeta| \to \infty$ the parts which are non-constant with respect to $\zeta$ vanish due to the Riemann-Lebesgue lemma for $|\Im \zeta|<\pi$; uniformity follows from $g(t) \exp (\pm \tfrac{t \Im \zeta}{\pi})$ being uniformly $L^1$-bounded in $|\Im \zeta| \leq \delta$ (see, e.g., proof of Thm.~IX.7 in \cite{RS75}). 

Next, according to Corollary~\ref{cor:prodminsol}, the minimal solution with respect to $S$ is given by
\begin{equation}\label{eq:minsolgenshg}
    F_{S,\mathrm{min}}(\zeta)= ( i \sh \tfrac{\zeta}{2})^{-s(\epsilon,n)} \prod_{k=1}^n F_{b_k,\mathrm{min}}(\zeta)
\end{equation}
with $s(+1,n) = 2\lfloor\tfrac{n}{2}\rfloor$ and $s(-1,n) = 2\lfloor\tfrac{n-1}{2}\rfloor$. For the stress-energy tensor at one-particle level, we obtain (using Corollary~\ref{cor:diagmass}, Lemma~\ref{lem:fgenform}, and Corollary~\ref{cor:minsolconjugation}) that
\begin{equation}
    F_2^{\mu\nu}(\zeta_1,\zeta_2+i\pi) =  G_\mathrm{free}^{\mu\nu}\left(\tfrac{\zeta_1+\zeta_2}{2}\right) F_q(\zeta_1-\zeta_2+i\pi), \quad F_q(\zeta) = q(\ch\zeta) F_{S,\mathrm{min}} (\zeta+i\pi)
\end{equation}
with $q$ a polynomial having real-valued coefficients and $q(-1)=1$.

Let $c:= 2^{s(\epsilon,n)-\deg q} \lvert c_q \rvert \prod_{k=1}^n F_{b_k,\mathrm{min}}^{\infty}$, where $c_q$ is the leading coefficient of $q$. By the preceding remarks we find that for some $c',c''$ with $0<c'< c < c''$ and $\delta, r > 0$:
\begin{equation}\label{eq:shgestimate}
    \forall |\Re\zeta| \geq r, |\Im \zeta|\leq \delta: \quad c' \leq \frac{|F_q(\zeta+i\pi)|}{\exp( (\deg q - \tfrac{1}{2}s(\epsilon,n)) |\Re\zeta|)} \leq c'' ,
\end{equation}
where $c'$ and $c''$ can be chosen arbitrarily close to $c$ for large enough $r$.

We can therefore conclude by Theorem~\ref{thm:qeimain} that a QEI of the form \eqref{eq:QEI} holds if $\deg q < \tfrac{1}{2} s(\epsilon,n) +1$ and cannot hold if $\deg q > \tfrac{1}{2} s(\epsilon,n)+1$. In case that $\deg q = \tfrac{1}{2} s(\epsilon,n)+1$, details of $q$ become relevant. This can only occur if $s(\epsilon,n)$ is even, i.e., $\epsilon = +1$. If here $c$ is less (greater) than $\tfrac{1}{4}$, then a QEI holds (cannot hold).

\subsection{(Generalized) Bullough-Dodd model}\label{sec:bd}
We now consider a class of integrable models which treat a single neutral scalar particle that is its own bound state. The presence of the bound state requires the S-function to have a specific ``bound state pole'' in the physical strip with imaginary positive residue and to satisfy a bootstrap equation for the self-fusion process. Such S-functions are classified in \cite[Appendix A]{CT15}. The Bullough-Dodd model itself (see \cite{AFZ79,FMS93} and references therein) corresponds to the maximally analytic element of this class which is given by $\zeta \mapsto S_{\mathrm{BD}}(\zeta;b) = S(\zeta;-\tfrac{2}{3}) S(\zeta;\tfrac{b}{3})S(\zeta;\tfrac{2-b}{3})$ where $b\in (0,1)$ is a parameter of the model. The full class allows for so-called CDD factors \cite{CDD56} and an exotic factor of the form $\zeta \mapsto e^{ia\sh \zeta}, a>0$.

In Lagrangian QFT, from a one-component field $\varphi$ and a Lagrangian
\begin{equation}
    \mathcal{L}_{\mathrm{BD}} = \tfrac{1}{2} \partial_\mu \varphi \partial^\mu \varphi - \frac{m^2}{6g^2} (2 e^{g \varphi} + e^{-2g \varphi})
\end{equation}
under the (perturbative) correspondence \makebox{$b = \frac{g^2}{2\pi} (1+\tfrac{g^2}{4\pi})^{-1}$}\cite{FMS93} one obtains as S-function $S_{\mathrm{BD}}(\cdot;b)$. For more general elements of the described class, no Lagrangian is known \cite{CT15}.

In our context, we will consider the generalized variant of the model, but for simplicity restrict to finitely many CDD factors and do not include the exotic factor:
\begin{defn}
    The \emph{generalized Bullough-Dodd model} is specified by the mass parameter $m>0$ and a finite sequence $(b_k)_{k\in\{1, \ldots ,n\}} \subset (0,1)+i\mathbb{R}, n\in\mathbb{N},$ which has an odd number of real elements and where the non-real $b_k$ appear in complex conjugate pairs. The one-particle little space is given by $\mathcal{K} = \mathbb{C}$, $\mathcal{G}=\{ e\}$, $V=1_\mathbb{C}$, and $M = m 1_{\mathbb{C}}$. $J$ corresponds to complex conjugation. The S-function $S_{\mathrm{gBD}}$ is of the form 
    \begin{equation}
        S_{\mathrm{gBD}}(\zeta) = S(\zeta;-\tfrac{2}{3}) \prod_{k=1}^n S(\zeta;\tfrac{b_k}{3})S(\zeta;\tfrac{2-b_k}{3}).
    \end{equation}
\end{defn}

Clearly, $S_{BD}$ is obtained from $S_{gBD}$ for $n=1$ and $b_1=b$. Since $S_{\mathrm{gBD}}$ is defined as a product of a finite number of factors of the form $S(\cdot;b)$, its minimal solutions exists and is given by, see Corollary~\ref{cor:prodminsol}, 
\begin{equation}\label{eq:fmingenbd}
    F_{\mathrm{gBD},\mathrm{min}}(\zeta) = (-i\sh \tfrac{\zeta}{2})^{-2n} F_{-2/3,\mathrm{min}}(\zeta) \prod_{k=1}^n F_{b_k/3,\mathrm{min}}(\zeta) F_{(2-b_k)/3,\mathrm{min}}(\zeta).
\end{equation}
It enters here that $S_{\mathrm{gBD}}(0) = -1$.

The presence of bound states in the model implies the presence of poles in the form factors of local operators \cite{BFK08}, in particular also for $F_2^{\mu\nu}$. For $F_1^{\mu\nu} \neq 0$ we expect a single first-order pole of $F_2^{\mu\nu}(\zeta,\zeta';x)$ at $\zeta'-\zeta = i\tfrac{2\pi}{3}$. In case that $F_1^{\mu\nu} = 0$ we expect $F_2^{\mu\nu}(\zeta,\zeta';x)$ to have no poles in $\mathbb{S}[0,\pi]$.

\begin{lem}[Stress tensor in the generalized BD model]
    A tensor-valued function $F^{\mu\nu}_2:\mathbb{C}^2\times \mathbb{M} \to \mathcal{K}^{\otimes 2}$ is a stress-energy tensor at one-particle level with respect to $S_{\mathrm{gBD}}$ and $\mathfrak{P} \subset \{ i \tfrac{2\pi}{3}\}$ iff it is of the form
    \begin{equation}\label{eq:tformbd}
        F_2^{\mu\nu}(\theta,\eta+i\pi) = G^{\mu\nu}_\mathrm{free}\left( \tfrac{\theta+\eta}{2}\right) e^{i(p(\theta;m)-p(\eta;m)).x} F_q(\eta-\theta+i\pi),
    \end{equation}
    with
    \begin{equation}\label{eq:tformbd2}
        F_q(\zeta) = q(\ch \zeta) (-2\ch \zeta -1)^{-1} F_{\mathrm{gBD},\mathrm{min}}(\zeta),
    \end{equation}
    where $F_{\mathrm{gBD},\mathrm{min}}$ is the unique minimal solution with respect to $S_{\mathrm{gBD}}$ and where $q$ is a polynomial with real coefficients and $q(-1)=1$.
\end{lem}

\begin{proof}
    By Theorem~\ref{thm:tformgen} and Corollary~\ref{cor:diagmass}, $F_2^{\mu\nu}$ is given by \eqref{eq:tformbd}, where $F:\mathbb{C}\to\mathbb{C}$ satisfies properties \ref{f2poles}--\ref{f2norm} of Theorem~\ref{thm:tformgen} with respect to $S_{\mathrm{gBD}}$. According to Lemma~\ref{lem:fgenform}, $F$ is of the form $F_q$ \eqref{eq:tformbd2}; the factor $(-2\ch \zeta -1)^{-1}$ takes the one possible first-order pole within $S[0,\pi]$, namely at $i \frac{2\pi}{3}$, into account. That $q$ has real coefficients is a consequence of property \ref{f2cpt} and Corollary~\ref{cor:minsolconjugation}.
    
    Conversely, it is clear that $F_2^{\mu\nu}$, respectively $F=F_q$, as given above has the properties \ref{f2poles}-\ref{f2norm}.
\end{proof}

\begin{thm}[QEI for the generalized BD model]\label{thm:qeibd}
    Let the stress-energy tensor at one-particle level be given by $F_2^{\mu\nu}$ as in Eq.~\eqref{eq:tformbd}. Then a QEI of the form
    \begin{equation}\label{eq:qeibd}
       \forall g\in \mathcal{S}_\mathbb{R}(\mathbb{R}) \, \exists c_g > 0 \, \forall \varphi \in \mathcal{D}(\mathbb{R},\mathcal{K}): \quad \braket{\varphi, T^{00}(g^2)\varphi} \geq -c_g \lVert \varphi \rVert_2^2
    \end{equation}
    holds if $\deg q < n+2$ and cannot hold if $\deg q > n+2$. In the case $\deg q = n+2$, introduce
    \begin{equation}c:= 2^{2n-\deg q} \lvert c_q\rvert F_{-2/3,\mathrm{min}}^\infty \prod_{k=1}^n F_{b_k/3,\mathrm{min}}^\infty F_{(2-b_k)/3,\mathrm{min}}^\infty,\end{equation}
    where $c_q$ denotes the leading coefficient of $q$. If here $c$ is less (greater) than $\tfrac{1}{4}$ then a QEI holds (cannot hold).
\end{thm}

\begin{proof}
    As the minimal solution $F_{\mathrm{gBD},\mathrm{min}}$ is given as a finite product of factors $\zeta\mapsto (-i \sh \tfrac{\zeta}{2})$ and $F_{b,\mathrm{min}}$, the asymptotic growth can be estimated analogously to the procedure in the introduction of Section~\ref{sec:examples}. Similar to the estimate \eqref{eq:shgestimate}, one obtains for some $c'$ and $c''$ with $0<c'< c < c''$ and some $\epsilon, r > 0$:
    \begin{equation}\label{eq:bdestimate}
        \forall |\Re\zeta| \geq r, |\Im \zeta|\leq \epsilon: \quad c' \leq \frac{|F_q(\zeta+i\pi)|}{\exp( (\deg q - n - 1) |\Re\zeta|)} \leq c'' ,
    \end{equation}
    where $c'$ and $c''$ can be chosen arbitrarily close to $c$ for large enough $r$.

    Noting that parity covariance is trivial for $\mathcal{K}=\mathbb{C}$ and applying Theorem~\ref{thm:qeimain} yields the conclusions from above depending on $\deg q$ and $c$.
\end{proof}

\subsection{Federbush model}\label{sec:fb}

The Federbush model is a well-studied integrable QFT model with a constant, but non-trivial, scattering function; see \cite{Fed61,STW76,Rui81,Rui82,CF01} and references therein. In Lagrangian QFT, the traditional Federbush model is described in terms of two Dirac fields $\Psi_1$, $\Psi_2$ by a Lagrangian density\footnote{The fields $\Psi_j$ take values in $\mathbb{C}^2$. $\epsilon_{\mu\nu}$ denotes the antisymmetric tensor with $\epsilon_{01} = -\epsilon_{10} = 1$. Other standard notations are $\bar\psi_j:= \psi_j^\dagger \gamma_0$ and $\slashed \partial = \gamma^\mu \partial_\mu$ with anticommuting matrices $\gamma^0,\gamma^1 \in \mathrm{Mat}(2\times 2,\mathbb{C})$, $[\gamma^\mu,\gamma^\nu]_+ = 2g^{\mu\nu}$.}
\begin{equation}
    \mathcal{L}_\mathrm{Fb} = \sum_{j=1}^2 \tfrac{1}{2} \bar\Psi_j (i \slashed\partial - m_j ) \Psi_j - \lambda \pi \epsilon_{\mu\nu} J_1^\mu J_2^\nu, \quad J_j^\mu = \bar\Psi_j \gamma^\mu \Psi_j.
\end{equation}
The Federbush model obeys a global $U(1)^{\oplus 2}$ symmetry since $\mathcal{L}_{\mathrm{Fb}}$ is invariant under
\begin{equation}\label{eq:u1symm}
    \Psi_j(x) \mapsto e^{2\pi i \kappa} \Psi_j(x),\quad \Psi_j^\dagger(x) \mapsto e^{-2\pi i \kappa} \Psi^\dagger_j(x),\quad \kappa \in \mathbb{R}, j=1,2.
\end{equation}
The stress-energy tensor of the model has been computed before \cite{SH78} and its trace (Eq.~(44) in the reference) is given by
\begin{equation}\label{eq:tracefbset}
    T^\mu_\mu = \sum_{j=1}^2 m_j \normord{\bar\Psi_j \Psi_j}
\end{equation}
which agrees with the (trace of the) stress-energy tensor of two free Dirac fermions. Note in particular that it is parity-invariant.

In our framework, the model can be described in the following way:
\begin{defn}
    The \emph{Federbush model} is specified by three parameters, the particle masses $m_1,m_2 \in (0,\infty)$ and the coupling parameter $\lambda \in (0,\infty)$. The symmetry group is $\mathcal{G}=U(1)^{\oplus 2}$. The one-particle little space is given by $L=(\mathcal{K},V,J,M)$ with $L = L_1 \oplus L_2$ and where for $j=1,2$ we define $\mathcal{K}_j=\mathbb{C}^2$ and
    \begin{equation}
       V_j(\kappa) = \left( \begin{matrix} e^{2\pi i \kappa} & 0 \\ 0& e^{-2\pi i \kappa} \end{matrix}\right), \quad J_j =  \left( \begin{matrix} 0 & -1 \\ -1 & 0\end{matrix}\right), \quad M_j = m_j \left( \begin{matrix} 1 & 0 \\ 0& 1\end{matrix}\right)
    \end{equation}
    as operators on $\mathcal{K}_j$ where $J_j$ is antilinear and for the choice of basis $\{ e_j^{(+)} \equiv (1,0)^t$, $e_j^{(-)} \equiv (0,1)^t \}$. 
    The S-function is denoted by $S_\mathrm{Fb} \in \mathcal{B}(\mathcal{K}^{\otimes 2})$. Its only nonvanishing components, enumerated as $\alpha,\beta =1+,1-,2+,2-$ corresponding to $e_{1/2}^{(\pm)}$, are given by $S_{\alpha\beta} := (S_\mathrm{Fb})_{\alpha\beta}^{\beta\alpha}$ with
    \begin{equation}
        S = - \left( \begin{matrix} 1 & 1 & e^{2\pi i \lambda} & e^{-2\pi i \lambda} \\ 1 & 1 & e^{-2\pi i \lambda} & e^{2\pi i \lambda} \\ e^{-2\pi i \lambda} & e^{2\pi i \lambda} & 1 & 1 \\ e^{2\pi i \lambda} & e^{-2\pi i \lambda} & 1 & 1 \end{matrix} \right).
    \end{equation}
\end{defn}

Note that $S_{\mathrm{Fb}}$ is a constant diagonal S-function; e.g., $S_{\alpha\beta}= S^\ast_{\beta\alpha} = S_{\beta\alpha}^{-1}$ imply that $S_{\mathrm{Fb}}$ is self-adjoint and unitary. Note also that, $S_{\alpha\beta} = S_{\bar \alpha \bar \beta} \neq S_{\beta\alpha}$, where $\bar{\alpha}$ corresponds to $\alpha \in \{ 1+,1-,2+,2-\}$ by flipping plus and minus. These relations correspond to the fact that $S_{\mathrm{Fb}}$ is C-, PT- and CPT- but not P- or T-symmetric. However, $S_{\mathrm{Fb}}$ has a P-invariant diagonal (in the sense of Eq.~\eqref{eq:pinvdiag}) due to $S_{\alpha\bar\alpha} = S_{\bar\alpha\alpha}$ (or Remark~\ref{rem:diagmodels}).

\begin{lem}[Stress tensor for the Federbush model]\label{thm:emfb}
    A tensor-valued function $F^{\mu\nu}_2:\mathbb{C}^2\times \mathbb{M} \to \mathcal{K}^{\otimes 2}$ is a stress-energy-tensor at one-particle level with respect to $S_{\mathrm{Fb}}$, is diagonal in mass (Eq.~\eqref{eq:diaginmass}), and has no poles, $\mathfrak{P}=\emptyset$, iff it is of the form
    \begin{equation}\label{eq:tformfb}
        F_2^{\mu\nu}(\theta,\eta+i\pi;x) = G_\mathrm{free}^{\mu\nu}\left(\tfrac{\theta+\eta}{2}\right) e^{iP(\theta,\eta+i\pi).x} F(\eta-\theta+i\pi)
    \end{equation}
    with
    \begin{equation}\label{eq:tformfb2}
        F(\zeta) = \sum_{j=1}^2 \left( -i \sh (\tfrac{\zeta}{2}) q^{\mathrm{s}}_j(\ch \zeta) \, e^{(+)}_j \otimes_{\mathrm{s}} e^{(-)}_j + \ch (\tfrac{\zeta}{2}) q_j^{\mathrm{as}}(\ch \zeta) \, e^{(+)}_j \otimes_{\mathrm{as}} e^{(-)}_j \right),
    \end{equation}
    for $e^{(+)}_j \otimes_{\mathrm{s/as}} e^{(-)}_j := e_j^{(+)} \otimes e_j^{(-)} \pm e_j^{(-)} \otimes e_j^{(+)}$ and where each $q^{\mathrm{s/as}}_j$ is a polynomial with real coefficients and $q^{\mathrm{s}}_j(-1)=1$.
    
    The stress-energy tensor at one-particle level is parity-covariant iff $q_1^{\mathrm{as}} = q_2^{\mathrm{as}} \equiv 0$.
\end{lem}

\begin{proof}
    By Theorem~\ref{thm:tformgen} and Corollary~\ref{cor:diagmass}, we have that Eq.~\eqref{eq:tformfb} holds with $F$ satisfying properties \ref{f2poles}-\ref{f2norm}. $U(1)^{\oplus 2}$-invariance, property \ref{f2ginv}, is equivalent to
    $$\forall \zeta\in \mathbb{C},\bkappa\in \mathbb{R}^2, r,s\in\{\pm\}, j,k\in \{1,2\}: \qquad \left( 1 - e^{2\pi i (r\kappa_j+s\kappa_k)}  \right) (e^{(r)}_j\otimes e^{(s)}_k , F(\zeta)) = 0.$$
    As a consequence, $(e^{(r)}_j\otimes e^{(s)}_k , F(\zeta)) = 0$ unless $j=k$ and $r=-s$. On the remaining components, $S$ acts like $-\mathbb{F}$; thus,
    \begin{equation}\label{eq:fbfrels}
        F(\zeta) = - \mathbb{F} F(-\zeta) = \mathbb{F}F(2i\pi - \zeta),
    \end{equation}
    which implies
    \begin{equation}
F(\zeta) = \sum_{j=1}^2 \left( -i \sh (\tfrac{\zeta}{2}) f^{\mathrm{s}}_j(\zeta) e_j^{(+)}\otimes_{\mathrm{s}} e_j^{(-)} + \ch (\tfrac{\zeta}{2}) f^{\mathrm{as}}_j(\zeta) e_j^{(+)}\otimes_{as} e_j^{(-)} \right) 
    \end{equation}
 for some functions $f_j^{\mathrm{s/as}}$, where we have factored out the necessary zeroes due to the relations \eqref{eq:fbfrels}. Then from the properties of $F$ we find $f_j^{\mathrm{s/as}}: \mathbb{C} \to \mathbb{C}$ to be analytic and to satisfy 
    \begin{equation}f_j^{\mathrm{s/as}}(\zeta) = f^{\mathrm{s/as}}_j(-\zeta) = f_j^{\mathrm{s/as}}(2\pi i -\zeta), \quad f^{\mathrm{s}}_j(i\pi) = 1,\end{equation}
    and $f^{\mathrm{as}}_j(i\pi)$ unconstrained. Moreover, $f_j^{\mathrm{s/as}}$ are regular in the sense of Eq.~\eqref{eq:fasest} of Lemma \ref{lem:fgenform}; the lemma implies that $f_j^{\mathrm{s/as}}(\zeta) = q_j^{\mathrm{s/as}}(\ch \zeta)$ with $q_j^{\mathrm{s}}(-1) = 1$.
    Since $J^{\otimes 2} F(\zeta+i\pi) = F(\bar\zeta+i\pi)$, $Je^{(\pm)}_j = -e_j^{(\mp)}$, and by the antilinearity of $J$, we find that $q^{\mathrm{s/as}}_j(\zeta+i\pi) = \overline{q_j^{\mathrm{s/as}}(\bar\zeta+i\pi)}$ such that $q_j^{\mathrm{s/as}}$ have real coefficients.
    
    Parity invariance of $F$, i.e., $\mathbb{F} F = F$, is equivalent to $q_j^{\mathrm{as}} = - q_j^{\mathrm{as}}$; thus, $q_j^{\mathrm{as}}=0$, because of $(\mathbbm{1} \mp \mathbb{F}) \, e_j^{(+)}\otimes_{\mathrm{s/as}} e_j^{(-)} = 0$.
\end{proof}

We see that the stress-energy tensor does not need to be parity-covariant. Concerning QEIs we state:

\begin{thm}[QEI for the Federbush model]\label{thm:qeifb}
    The parity-covariant part of the stress-energy tensor at one-particle level, given by $F_2$ in Eq.~\eqref{eq:tformfb} with $q_1^{\mathrm{as}} = q_2^{\mathrm{as}} \equiv 0$, satisfies a one-particle-QEI of the form 
    \begin{equation}\label{eq:1partqeifb}
        \forall g\in \mathcal{S}_\mathbb{R}(\mathbb{R}) \, \exists c_g > 0 \,\forall \varphi \in \mathcal{D}(\mathbb{R},\mathcal{K}): \quad \braket{\varphi, T_P^{00}(g^2)\varphi} \geq -c_g \lVert \varphi \rVert_2^2 \quad 
    \end{equation}
    iff $q^{\mathrm{s}}_1=q^{\mathrm{s}}_2 \equiv 1$.
    
    The candidate stress-energy tensor given by Eq.~\eqref{eq:tmunu2} (i.e. for $q^{\mathrm{s}}_1 = q^{\mathrm{s}}_2 = 1, q^{\mathrm{as}}_1= q^{\mathrm{as}}_2=0$) satisfies a QEI of the form
    \begin{equation}\label{eq:qeifb}
        T^{00}(g^2) \geq - \left(\sum_{j=1}^2 \frac{m_j^3}{2\pi^2} \int_1^\infty ds \lvert\widetilde{g}(m_j s)\rvert^2 w_-(s) \right) \mathbbm{1}
    \end{equation}
    with $w_-(s) = s\sqrt{s^2-1} - \log( s + \sqrt{s^2-1})$ and in the sense of a quadratic form on $\mathcal{D}\times \mathcal{D}$. 
\end{thm}

\begin{proof}
    In case that one of the $q^{\mathrm{s}}_j$ is different from $1$ we have for some $c, r > 0$ that $|q^{\mathrm{s}}_j(\ch \zeta) \sh \tfrac{\zeta}{2}| \geq c e^{3 |\Re \zeta|/2}$ for all $|\Re \zeta| \geq r$. Therefore, no QEI can hold due to Theorem~\ref{thm:qeimain}\ref{item:noqei} and Remark~\ref{rem:noparitycovariance} with $u=e_j^{(+)} \pm e_j^{(-)}$ for some $j\in \{1,2\}$. For $q^{\mathrm{s}}_1=q^{\mathrm{s}}_2 \equiv 1$ (and $q^{\mathrm{as}}_1=q^{\mathrm{as}}_2 \equiv 0$), Theorem~\ref{thm:qeimain}\ref{item:qei} yields Eq.~\eqref{eq:1partqeifb}. In that case $F(\zeta) = (- i \sh \tfrac{\zeta}{2} ) I_{\otimes 2}$ which coincides with the expression in \eqref{eq:constscatttensor1} due to $P_+ I_{\otimes 2} = 0$ and $P_- I_{\otimes 2} = I_{\otimes 2}$. Since $S_{\mathrm{Fb}}$ is constant and diagonal by Remark~\ref{rem:diagmodels}, Theorem~\ref{thm:qeiconstS} applies and yields Eq.~\eqref{eq:qeifb}.
\end{proof}

We see that for the Federbush model, requiring a one-particle QEI fixes a unique (parity-covariant part of the) stress-energy tensor at one-particle level that extends---since $S_\mathrm{Fb}$ is constant---to a dense domain of the full interacting state space. The parity-covariant part is in agreement with preceding results for the stress-energy tensor at one-particle level \cite[Sec.~4.2.3]{CF01}. This indicates that the parity-violating part of our expression is indeed not relevant for applications in physics. Our candidate for the full stress-energy tensor has the same trace as in \cite{SH78}. That the respective energy density satisfies a generic QEI is no surprise after all, as the QEI results are solely characterized in terms of the trace of the stress-energy tensor which here agrees with that of two free Dirac fermions (as was indicated also by Eq.~\eqref{eq:tracefbset}).

\subsection{O(n)-nonlinear sigma model}\label{sec:onsigma}

The $O(n)$-nonlinear sigma model is a well-studied integrable QFT model of $n$ scalar fields $\phi_j, j=1, \ldots ,n$, that obey an $O(n)$-symmetry. For a review see \cite[Secs.~6--7]{AAR01} and references therein. In Lagrangian QFT, it can be described by a combination of a free Lagrangian and a constraint
\begin{equation}
    \mathcal{L}_\mathrm{NLS} = \tfrac{1}{2} \partial_\mu \Phi^t \partial^\mu \Phi, \quad \Phi^t \Phi = \frac{1}{2g}, \quad \Phi = (\phi_1,\dots,\phi_n)^t,
\end{equation}
where $g \in (0,\infty)$ is a dimensionless coupling constant. Clearly, $\mathcal{L}_\mathrm{NLS}$ is invariant for $\Phi$ transforming under the vector representation of $O(n)$, i.e.,
\begin{equation}
    \Phi(x) \mapsto O \Phi(x), \quad O \in \mathrm{Mat}_\mathbb{R}(n\times n), \quad O^t = O^{-1}.  
\end{equation}
Note that the model---other than one might expect naively from $\mathcal{L}_\mathrm{NLS}$---describes massive particles. This is known as dynamical mass transmutation; the resulting mass of the $O(n)$-multiplet can take arbitrary positive values depending on a choice of a mass scale and corresponding renormalized coupling constant; see, e.g., \cite[Sec.~7.2.1]{AAR01} and \cite{JN88}.

In our framework, the model can be described in the following way:

\begin{defn}\label{defn:onnls}
    The \emph{$O(n)$-nonlinear sigma model} is specified by two parameters, the particle number $n \in \mathbb{N}$, $n \geq 3,$ and the mass $m > 0$. The one-particle little space $(\mathcal{K},V,J,M)$ is given by $\mathcal{K}=\mathbb{C}^n$ with the defining/vector representation $V$ of $\mathcal{G}=O(n)$, $M = m\mathbbm{1}_{\mathbb{C}^n}$, and where $J$ is complex conjugation in the canonical basis of $\mathbb{C}^n$. The S-function is given by
    \begin{equation}\label{eq:smatrixnls}
        S_{\mathrm{NLS}}(\zeta) := (b(\zeta) \mathbbm{1} + c(\zeta) \mathbb{F} + d(\zeta) \mathbb{K})\mathbb{F},
    \end{equation}
    where in the canonical basis of $\mathbb{C}^n$
    \begin{equation}
        \mathbbm{1}^{\gamma\delta}_{\alpha\beta} = \delta_\alpha^\gamma \delta_\beta^\delta, \quad \mathbb{F}_{\alpha\beta}^{\gamma\delta} = \delta_\alpha^\delta \delta_\beta^\gamma, \quad \mathbb{K}_{\alpha\beta}^{\gamma\delta} = \delta^{\gamma\delta} \delta_{\alpha\beta},\quad \alpha,\beta,\gamma,\delta = 1, \ldots ,n,
    \end{equation}
    \begin{equation}
        b(\zeta) = s(\zeta)s(i\pi-\zeta), \quad c(\zeta) = -i\pi\nu\zeta^{-1} b(\zeta),\quad d(\zeta)=-i\pi\nu(i\pi-\zeta)^{-1}b(\zeta),
    \end{equation}
    and
    \begin{equation}
        \nu = \tfrac{2}{n-2}, \quad s(\zeta) = \frac{\Gamma\left(\frac{\nu}{2} + \frac{\zeta}{2\pi i}\right)\Gamma\left(\frac{1}{2} + \frac{\zeta}{2\pi i}\right)}{\Gamma\left(\frac{1+\nu}{2} +\frac{\zeta}{2\pi i}\right) \Gamma\left(\frac{\zeta}{2\pi i}\right)} .
    \end{equation}
\end{defn}
$S_{\mathrm{NLS}}$ is the unique maximally analytic element of the class of $O(n)$-invariant S-functions \cite{ZZ78}. Maximal analyticity means here that in the physical strip $\mathbb{S}(0,\pi)$, the S-function has no poles and the minimal amount of zeroes which are compatible with the axioms for an S-function, i.e., \ref{sunit}--\ref{sinv}.
Its eigenvalue decomposition is given by
\begin{equation}\label{eq:nlsedecomp}
    S_{\mathrm{NLS}}(\zeta) = \left(S_+(\zeta) \tfrac{1}{2} \left( \mathbbm{1}+\mathbb{F}-\tfrac{2}{n} \mathbb{K}\right) + S_-(\zeta) \tfrac{1}{2} \left( \mathbbm{1} - \mathbb{F}\right) + S_0(\zeta) \tfrac{1}{n}\mathbb{K}\right) \mathbb{F},
\end{equation}
with $S_\pm = b \pm c$ and $S_0 = b+c+nd$. The S-function is P-, C-, and T-symmetric and satisfies $S_{\mathrm{NLS}}(0) = - \mathbb{F}$.

As a first step, we establish existence of the minimal solution with respect to $S_0$ and an estimate of its asymptotic growth:

\begin{lem}\label{lem:minsolnls0}
    The minimal solution with respect to $S_0$ exists and is given by
    \begin{equation}\label{eq:charf0}
        F_{0,\mathrm{min}}(\zeta) = (-i \sh \tfrac{\zeta}{2}) F_{f_0}(\zeta), \quad f_0(t) = \frac{e^{-t} + e^{-\nu t}}{e^t + 1}.
    \end{equation}
    Moreover, there exist $0< c \leq c'$, $r>0$ such that
    \begin{equation}\label{eq:f0asympt}
        \forall |\Re\zeta| \geq r, \Im \zeta \in [0,2\pi]: \quad c \leq \frac{|F_{0,\mathrm{min}}(\zeta)|}{|\Re \zeta |^{-(1+\frac{\nu}{2})} \exp|\Re \zeta|} \leq c'.
    \end{equation}
\end{lem}

\begin{proof}
    The characteristic function $f_0= f[-S_0]$ is computed in Appendix~\ref{app:charfct}. Clearly, it is smooth and exponentially decaying. Applying Lemma~\ref{lem:minsol} (uniqueness) and Theorem~\ref{thm:propsofintreps} (existence) we find that $F_{f_0}$ is well-defined and that $F_{0,\mathrm{min}}$ exists and agrees with the expression claimed. 
    The estimate of Eq.~\eqref{eq:fminasgrowth} together with 
    \begin{equation}f_0(t) = 1 -(1+\tfrac{\nu}{2}) t + \mathcal{O}(t^2), \quad t\to 0 \end{equation}
    and the estimate
    \begin{equation}\forall |\Re \zeta| \geq r>0 : \quad (1-e^{-2r}) \exp{|\Re \zeta|} \leq | 2\sh \zeta | \leq (1+e^{-2r}) \exp{|\Re \zeta|}\end{equation}
    imply \eqref{eq:f0asympt}.
\end{proof}

\begin{lem}[Stress tensor in NLS model]\label{thm:emnls}
    A tensor-valued function $F^{\mu\nu}_2:\mathbb{C}^2\times \mathbb{M} \to \mathcal{K}^{\otimes 2}$ is a parity covariant stress-energy tensor at one-particle level with respect to $S_{\mathrm{NLS}}$ with no poles, $\mathfrak{P}=\emptyset$, iff it is of the form
    \begin{equation}\label{eq:tformnls}
        F_2^{\mu\nu}(\theta,\eta+i\pi;x) = G_\mathrm{free}^{\mu\nu}\left(\tfrac{\theta+\eta}{2}\right) e^{i(p(\theta;m)-p(\eta;m)).x} F(\eta-\theta+i\pi),
    \end{equation}
    with
    \begin{equation}\label{eq:tformnls2}
        F(\zeta) = q(\ch \zeta) F_{0,\mathrm{min}} (\zeta) I_{\otimes 2},
    \end{equation}
    where $F_{0,\mathrm{min}}$ is the unique minimal solution with respect to the S-matrix eigenvalue $S_0$ and $q$ is a polynomial with real coefficients with $q(-1)=1$.
\end{lem}

\begin{proof}
    By Corollary~\ref{cor:diagmass}, $F_2^{\mu\nu}$ has the form \eqref{eq:tformnls} with $F$ satisfying properties \ref{f2poles}-\ref{f2norm} in Theorem~\ref{thm:tformgen}. By \ref{f2ginv}, $F(\zeta)$ is an $O(n)$-invariant 2-tensor for each $\zeta$. The general form of such a tensor is $F(\zeta) = \lambda(\zeta) I_{\otimes 2}$ with $\lambda:\mathbb{C}\to \mathbb{C}$ \cite[Sec.~4, case \lParent a\rParent]{ADO87}.
    
    Consider now property \ref{f2ssym}, $F(\zeta) = S(\zeta) F(-\zeta)$ for $S= S_{\mathrm{NLS}}$. Taking the scalar product of both sides with $\tfrac{1}{n} I_{\otimes 2}$ in $(\mathbb{C}^n)^{\otimes 2}$ yields
    \begin{equation}\lambda(\zeta) = \tfrac{1}{n} (I_{\otimes 2}, S_{\mathrm{NLS}}(-\zeta)I_{\otimes 2}) \lambda(-\zeta) = S_0(-\zeta) \lambda(-\zeta)\end{equation}
    by Eq.~\eqref{eq:smatrixnls} and $\mathbbm{1}I_{\otimes 2} = \mathbb{F}I_{\otimes 2} = \tfrac{1}{n} \mathbb{K}I_{\otimes 2}$. Here we used that $\mathbb{F}I_{\otimes 2} = J^{\otimes 2}I_{\otimes 2} = I_{\otimes 2}$ by Remark~\ref{rem:identityelement}.
    
    In summary, Lemma~\ref{lem:fgenform} can be applied to $\lambda$, so that $\lambda(\zeta) = q(\ch(\zeta))F_{0,\mathrm{min}}(\zeta)$, and thus, $F$ has the form \eqref{eq:tformnls2}. That $q$ has real coefficients is a consequence of \ref{f2cpt} and Corollary~\ref{cor:minsolconjugation}.
    
    Conversely, it is clear that $F_2^{\mu\nu}$ as in Eq.~\eqref{eq:tformnls}, respectively $F$, has the properties \ref{f2poles}--\ref{f2norm}.
\end{proof}

\begin{thm}[QEI for the NLS model]\label{thm:qeinls}
    The stress-energy tensor at one-particle level given by $F_2$ in Eq.~\eqref{eq:tformnls} satisfies
    \begin{equation}\label{eq:qeinls}
        \forall g\in \mathcal{S}_\mathbb{R}(\mathbb{R}) \, \exists c_g > 0 \,\forall \varphi \in \mathcal{D}(\mathbb{R},\mathcal{K}): \quad \braket{\varphi, T^{00}(g^2)\varphi} \geq -c_g \lVert \varphi \rVert_2^2 \quad 
    \end{equation}
    iff $q \equiv 1$.
\end{thm}

\begin{proof}
    Given $F_2$ as in Lemma~\ref{thm:emnls} and using $\widehat{I_{\otimes 2}} = \mathbbm{1}_{\mathcal{K}}$, we have $\lVert \hat{F}(\zeta) \rVert_{\mathcal{B}(\mathcal{K})} = \lvert q(\ch\zeta) F_{0,\mathrm{min}}(\zeta)\rvert$. Thus, by Lemma~\ref{lem:minsolnls0} there exist $r>0$ and $0< c \leq c'$ such that
    \begin{equation}\forall \zeta \in |\Re \zeta| > r, \Im \zeta \in [0,2\pi]: \quad c \, t(\zeta) \exp |\Re \zeta| \leq \lVert \hat{F}(\zeta) \rVert_{\mathcal{B}(\mathcal{K})} \leq c^\prime t(\zeta) \exp |\Re \zeta|\end{equation}
    with $t(\zeta) = |\Re \zeta|^{-(1+\frac{\nu}{2})} |q(\ch \zeta)|$. Note that for $q \equiv 1$, $t(\zeta)$ is polynomially decaying, whereas for non-constant $q$, $t(\zeta)$ is exponentially growing. Thus, if $q$ is constant ($q \equiv 1$), we have $c^\prime t(\zeta) < \tfrac{1}{4}$ for large enough $|\Re \zeta|$; and if $q$ is not constant, then $c t(\zeta) > \tfrac{1}{4}$ for large enough $|\Re \zeta|$. We conclude by Theorem~\ref{thm:qeimain} that a QEI of the form \eqref{eq:qeinls} holds iff $q \equiv 1$.
\end{proof}

\newpage

\section{Conclusion and outlook}\label {sec:conclusion}

We have established QEIs in a larger class of 1+1d integrable models than previously known in the literature. In particular, QEIs for generic states hold in a wide class of models with \emph{constant} scattering functions, including not only the Ising model, as known earlier, but also the Federbush model. Moreover, the class includes combinations and bosonic or fermionic variants of these models. In all of these situations, the form factor $F_2$ of the energy density determines the entire operator.

Furthermore, we have established necessary and sufficient conditions for QEIs to hold \emph{at one-particle level} in generic models, which may include bound states or several particle species. Also in this case, only $F_2$ contributes to expectation values of the energy density, and the conditions for QEIs are based on the large-rapidity behaviour of $F_2$.
At the foundation of both results was a characterization by first principles of the form of the energy density. However, we found that those principles do not constrain polynomial prefactors (in $\ch \zeta$) added to a viable candidate for the energy density (at one-particle level). As seen in the case of the Bullough-Dodd, the Federbush, and the $O(n)$-nonlinear sigma model, one-particle QEIs can then fix the energy density at one-particle level partially or entirely, in analogy to \cite{BC16}.

Our results suggest a number of directions for further investigation, of which we discuss the most relevant ones:

\paragraph{What is the nature of the freedom in the form of the stress-energy tensor?}
    The factor $q(\ch \zeta)$ in the energy density was partially left unfixed by our analysis.  
    At least in the scalar case ($\mathcal{K}= \mathbb{C}$), it can be understood as a polynomial in the differential operator $\square = g^{\mu\nu} \partial_\mu \partial_\nu$ acting on $T^{\mu\nu}$: Given a stress-energy tensor $T^{\mu\nu}$, define $\tilde{T}^{\mu\nu} := q(-1-\tfrac{\square}{2\,M^2}) T^{\mu\nu}$ for some polynomial $q$. Then at one-particle level
    $$F_2^{[\tilde{T}^{\mu\nu}(x)]}(\bzeta) = q(\ch (\zeta_1-\zeta_2)) F_2^{[T^{\mu\nu}(x)]}(\bzeta),$$
    and, provided that $q(-1) = 1$, $F_2^{[\tilde{T}^{\mu\nu}(x)]}$ defines another valid candidate for the stress-energy tensor at one-particle level.
    However, for generic models, $q$ may depend on the particle types and cannot be understood in terms of derivatives only.
    
    In the physics literature, given a concrete model, a few standard methods exist to check the validity of a specific choice of $q$: In case the model admits a Lagrangian, perturbation theory checks are used, e.g., \cite{BK02,BFK10,BFK13}. In case the model can be understood as a perturbation of a conformal field theory model, a scaling degree for the large-rapidity behaviour (conformal dimension) of the stress-energy tensor can be extracted, which fixes the large-rapidity behaviour of $F_2$, e.g., \cite{Zam86,DSC96,CF01}. The large-rapidity scaling degree is also related to momentum-space clustering properties, which were studied for some integrable models, e.g., \cite{Smi92,KM93,BFK21}. But in the general case, none of these methods may be available, and other constraints---perhaps from QEIs in states of higher particle number---might need to take their place.
    
\paragraph{Which other models can be treated with these methods?}
    We performed our analysis of one-particle QEIs in a very generic setting; there are nevertheless some limitations.
    For one, we employed the extra assumption of parity covariance of the stress-energy tensor. While parity invariance of the scattering function (and therefore covariance of the stress-energy-tensor) is satisfied in many models, it is not fully generic. Nevertheless, a non-parity-covariant stress-energy tensor is still subject to constraints by our results; in particular, the necessary condition we gave for a one-particle QEI to hold remains unmodified (see Remark~\ref{rem:noparitycovariance}). We expect a sufficient condition for a one-particle QEI, similar to the one presented in Theorem~\ref{thm:qeimain}\ref{item:qei}, to apply also in a parity-breaking situation. Some numerical tests indicate this; however, an analytic proof remained elusive to us.
    
    Another point is the decomposition of the two-particle form factor of the (trace of the) stress-energy tensor $F$ into polynomials and factors which are fixed by the model (including the minimal solutions and pole factors). For generic models, multiple polynomial prefactors can appear (at least one for each eigenvalue of the S-function). In typical models, these are few to begin with, and symmetries exclude many of those prefactors (as was presented for the Federbush or the $O(n)$-nonlinear sigma model). In other situations, however, there might be too many unfixed factors for the QEI to meaningfully constrain them.
    
    Lastly, we should remark that also in the presence of higher-order poles in the scattering function, the poles in the form factors are expected to be of first-order \cite{BK02,BFK06} so that such models should be tractable with our methods. This includes for instance the $Z(n)$-Ising, sine-Gordon, or Gross-Neveu model. Also generic Toda field theories don't seem to pose additional problems.
    
\paragraph{Do QEIs hold in states with higher particle numbers?}
    Apart from the case of constant S-functions, we treated only one-particle expectation values of the energy density in this paper. At $n$-particle level, generically the form factors $F_{1}, ..., F_{2n}$ all enter the expectation values; these are more challenging to handle since the number of rapidity arguments increases and since additional (``kinematic'') poles arise at the boundary of the analyticity region that were absent in the case $n=1$. Therefore, treating higher particle numbers requires new methods; a result in either direction, validity or non-validity of QEIs, is by no means straightforward and we leave this analysis to future work. While we conducted a few promising numerical tests for the sinh-Gordon model at two-particle level \cite{Man23b}, these can only serve as an indication. We do not expect to obtain numerical results at much higher particle numbers due to computational complexity scaling exponentially with $n$.

\newpage
\appendix 

\section{The minimal solution}\label{app:fmin}

This appendix collects rigorous results on existence and uniqueness of minimal solutions for integrable models, as well as estimates for their asymptotic growth, which are central to the question whether a QEI holds in the model (see Sec.~\ref{sec:onepQEI}). 

Some of these results are also contained in \cite{BC16}, whereas a less rigorous but informative treatment can be found in \cite{KW78}. Our existence result is based on an integral representation of the minimal solution which is well-known in principle and has been employed before in many concrete models, e.g., sinh-Gordon \cite{FMS93}, $SU(N)$-Gross-Neveu \cite{BFK10}, and $O(N)$-nonlinear-$\sigma$ \cite{BFK13}. Existence of the integral representations was argued in \cite{KW78}, but without giving explicit assumptions. General results on the asymptotic growth of the minimal solution, based on this integral representation, are new to the best of the authors' knowledge.

\subsection{Eigenvalue decomposition of the S-function
}
To begin with, we establish the eigenvalue decomposition of an S-function. Since \makebox{$S(\theta)\in\mathcal{B}(\mathcal{K}^{\otimes 2})$} is unitary for real arguments, it is diagonalizable; this extends to complex arguments by analyticity:

\begin{prop}\label{prop:sedecomp}
    Let $S$ be an S-function and $D(S)$ its domain of analyticity. Then there exists \makebox{$k\in \mathbb{N}$} and a discrete set $\Delta(S) \subset D(S)$ such that the number of distinct eigenvalues of $S(\zeta)$ is $k$ for all $\zeta \in D(S) {\setminus} \Delta(S)$ and strictly less than $k$ for all $\zeta \in \Delta(S)$. Further, for any simply connected domain $\mathcal{D} \subset D(S) \setminus \Delta(S)$ there exist analytic functions $S_i:\mathcal{D} \to \mathbb{C}$, and analytic projection-valued functions $P_i: \mathcal{D} \to \mathcal{B}(\mathcal{K}^{\otimes 2})$, $i=1, \ldots ,k$ with
    \begin{equation}\label{eq:evdecomp}
        S(\zeta) = \sum_{i=1}^k S_i(\zeta) P_i(\zeta), \quad \zeta \in \mathcal{D}
    \end{equation}
    such that for all $\zeta \in \mathcal{D}$,
    \begin{enumerate}[label=(\alph*)]
     \item $S_1(\zeta), \ldots ,S_k(\zeta)$ coincide with the eigenvalues of $S(\zeta)$ and $P_1(\zeta), \ldots ,P_k(\zeta)$ coincide with the projectors onto the respective eigenspaces.\\ In particular, $P_i(\zeta) P_j(\zeta) = \delta_{ij} P_i(\zeta)$ for $i,j=1, \ldots ,k$,
     \item if $-\zeta \in \mathcal{D}$ one has $S_i(-\zeta) = S_i(\zeta)^{-1}$ and $P_i(-\zeta) = P_i(\zeta)$,
     \item if $\bar\zeta\in \mathcal{D}$ one has $\overline{S_i(\bar\zeta)} = S_i(\zeta)^{-1}$ and $P_i(\bar\zeta) = P_i(\zeta)^\dagger$,
     \item each $P_i$ satisfies $\mathcal{G}$-invariance, $[P_i(\zeta), V(g)^{\otimes 2}] = 0, g \in \mathcal{G}$, CPT-invariance, $P_i(\zeta) = J^{\otimes 2} \mathbb{F} P_i(\zeta)^\dagger \mathbb{F} J^{\otimes 2},$ and translational invariance, $(E_m\otimes E_{m'})P_i(\zeta) = P_i(\zeta) (E_{m'}\otimes E_{m})$ for all $m,m'\in\mathfrak{M}$.
    \end{enumerate}
    The decomposition is unique up to relabeling.
\end{prop}

\begin{proof}
    For the eigenvalue decomposition of a matrix-valued analytic function, see e.g. \cite[Theorem~4.8]{Par78} or \cite[Chapter 2]{Kat95}. Restricting $S$ to its domain of analyticity $\mathcal{D}(S)$ we can apply the theorem from the first-named reference: For some $k\in\mathbb{N}$ and a discrete set $\Delta(S) \subset D(S)$ and any simply connected domain $\mathcal{D}\subset D(S)\setminus \Delta(S)$ we obtain pairwise distinct analytic functions $S_i: \mathcal{D} \to \mathbb{C}$ and $P_i,D_i:\mathcal{D}\to \mathcal{B}(\mathcal{K}^{\otimes 2})$ for $i=1,\ldots,k$ such that for each $\zeta \in \mathcal{D}$
    $$S(\zeta)= \sum_{i=1}^k S_i(\zeta) P_i(\zeta) + D_i(\zeta)$$
    is the unique Jordan decomposition of $S(\zeta)$ with eigenvalues $S_i(\zeta)$, eigenprojectors $P_i(\zeta)$ and eigennilpotents $D_i(\zeta)$, $i=1, \ldots ,k$. Let us enlarge $\mathcal{D}$ to $\tilde{\mathcal{D}}$ within $D(S)\setminus \Delta(S)$ such that $\tilde{\mathcal{D}}\cap \mathbb{R} \subset \mathbb{R}$ is open and non-empty and such that $\tilde{\mathcal{D}}$ is still simply connected; this is always possible since $\mathbb{C}\setminus D(S)$ and $\Delta(S)$ are discrete, i.e., countable and without finite accumulation points. Since $S(\theta)$ for $\theta \in \mathbb{R}$ is unitary and therefore semisimple we find that $D_i\restriction {\tilde{\mathcal{D}}\cap\mathbb{R}} = 0$. Since $D_i$ is analytic, this implies $D_i=0$. From the properties of the Jordan decomposition we further infer that $P_i(\zeta)P_j(\zeta)= \delta_{ij}P_i(\zeta), \,i,j=1, \ldots ,k$.
    
    The properties (b)--(d) are implied by the corresponding properties of $S$.
\end{proof}

Note that the $S_i$ (within any domain $\mathcal{D}$ from above) satisfy all the properties of a scalar S-function except for crossing symmetry. Specifically, these are the properties \ref{sunit} and \ref{sherm}, since \ref{scpt}, \ref{sybeq}, \ref{strans}, and \ref{sinv} are trivially satisfied in the scalar setting.

In typical examples, the decomposition in Eq.~\eqref{eq:evdecomp} can be extended to all of $\mathbb{C}$ if one allows for meromorphic $S_i$ and $P_i$. This applies particularly to models with constant eigenprojectors (e.g., all models with constant or diagonal S-functions) but also the other examples treated in Section~\ref{sec:examples}.

\subsection{Uniqueness of the minimal solution and decomposition of one-particle solutions}\label{sec:fminunique}

Throughout the remainder of Appendix \ref{app:fmin}, we intend to analyse eigenvalues $S$ of some matrix-valued S-function; thus, $S$ will denote a $\mathbb{C}$-valued (not matrix-valued) function from now on. The content of the present section is taken from \cite{BC16} with slight generalizations. Central to the section is:

\begin{lem}\label{lem:minsol}
    Let $S: \mathbb{C}\to \mathbb{C}$ be a meromorphic function with no poles on the real line. Then there exists at most one meromorphic function $F:\mathbb{C}\to \mathbb{C}$ such that
    \begin{enumerate}[label=(\alph*),ref=(\alph*)]
     \item \label{fzeroes} $F$ has no poles and no zeroes in $\mathbb{S}[0,\pi]$, except for a first-order zero at $0$ in case that $S(0)=-1$,
     \item \label{fest} $\exists a,b,r>0\, \forall |\Re \zeta| \geq r, \Im \zeta \in [0,\pi] : \quad \lvert \log \lvert F(\zeta)\rvert \rvert \leq a + b |\Re \zeta|$,
     \item \label{fperiod} $F(i\pi + \zeta) = F(i\pi - \zeta)$,
     \item \label{fssym} $F(\zeta) = S(\zeta) F(-\zeta)$,
     \item \label{fnorm} $F(i\pi)=1$.
    \end{enumerate}
\end{lem}

If such a function exists, we will refer to it as \emph{the minimal solution} $F_{S,\mathrm{min}}$ with respect to $S$. Due to \ref{fssym}, a necessary condition for existence is the relation $S(-\zeta) = S(\zeta)^{-1}$ for all $\zeta \in \mathbb{C}$.

\begin{proof}[Proof of Lemma~\ref{lem:minsol}]
    Assume that there are two functions $F_A$, $F_B$ with the stated properties. Then the meromorphic function $G(\zeta):=F_A(\zeta)/F_B(\zeta)$ has neither poles nor zeroes in $\mathbb{S}[0, 2 \pi]$ and satisfies $G(\zeta) = G(-\zeta) = G(\zeta+2\pi i)$. These relations imply that $q:=G\circ \ch^{-1}$ is well-defined and entire. The asymptotic estimates \ref{fest} for $\lvert \log \lvert F_{A/B} \rvert \rvert$ imply an analogous estimate for $\lvert \log \lvert G\rvert \rvert = \lvert \log \lvert F_A\rvert - \log \lvert F_B \rvert \rvert$ by the triangle inequality. Thus $q$ is polynomially bounded at infinity and therefore a polynomial. However, since $q$ does not have zeroes, it must be a constant with $q\equiv q(-1)=1$ due to $G(i\pi)=1$. Hence $F_A=F_B$.
\end{proof}

\begin{cor}\label{cor:minsolconjugation}
    If in addition $\overline{S(\bar\zeta)} = S(\zeta)^{-1}, \zeta \in \mathbb{C},$ then it holds that
    \begin{equation}
        F_\mathrm{min}(\zeta) = \overline{F_\mathrm{min}(-\bar\zeta)}.
    \end{equation}
\end{cor}

\begin{proof}
    Since $\overline{S(-\bar\zeta)} = S(\zeta)$, it is clear that $\zeta\mapsto \overline{F_\mathrm{min}(-\bar\zeta)}$ satisfies the same properties \ref{fzeroes}--\ref{fnorm} as $F_\mathrm{min}$. By uniqueness they have to be equal.
\end{proof}

\begin{cor}\label{cor:prodminsol}
    For $n\in\mathbb{N}$ let $S_1, \ldots ,S_n:\mathbb{C}\to\mathbb{C}$ be meromorphic functions such that their minimal solutions $F_{j,\mathrm{min}}$ exist.
    Then the minimal solution with respect to $\zeta \mapsto S_\Pi(\zeta) = \prod_{j=1}^n S_j(\zeta)$ exists and is given by
    \begin{equation}\label{eq:fminprod}
        \zeta \mapsto F_{\Pi,\mathrm{min}}(\zeta) = (-i\sh \tfrac{\zeta}{2})^{-2\lfloor\tfrac{s}{2}\rfloor} \prod_{j=1}^n F_{j,\mathrm{min}}(\zeta),
    \end{equation}
    where $s = |\{ j : S_j(0) = -1\}|$.
\end{cor}

\begin{proof}
    One easily checks that Eq.~\eqref{eq:fminprod} satisfies conditions \mbox{\ref{fest}--\ref{fnorm}} of Lemma~\ref{lem:minsol} with respect to $S_{\Pi}$.
    Also, counting the order of zeroes at $0$ on the r.h.s. yields $s - 2 \lfloor \tfrac{s}{2} \rfloor$, which evaluates to $1$ for odd $s$ (when $S_{\Pi}(0) = -1$) and to $0$ otherwise (when $S_{\Pi}(0) = +1$), thus establishing condition \ref{fzeroes}.
\end{proof}

We now apply theses results to classify ``non-minimal'' solutions, having more zeroes or poles than allowed by condition \ref{fzeroes}:

\begin{lem}\label{lem:fgenform}
    Let $F:\mathbb{C}\to\mathbb{C}$ be a meromorphic function which satisfies properties \ref{fperiod}-\ref{fnorm} of Lemma \ref{lem:minsol} with respect to some meromorphic function $S$, and suppose
    \begin{equation}\label{eq:fasest}
        \exists a,b,r > 0\, \forall |\Re \zeta|\geq r, \Im \zeta \in [0,\pi]: \quad |F(\zeta)| \leq a \exp b\lvert \Re\zeta\rvert.
    \end{equation}
    Assume further that the minimal solution $F_{S,\mathrm{min}}$ with respect to $S$ exists. Then there is a unique rational function $q$ with $q(-1)=1$ such that
    \begin{equation}\label{eq:fgenform}
        F(\zeta) = q(\ch\zeta) F_{S,\mathrm{min}}(\zeta).
    \end{equation}  
    In particular, if $F$ has no poles in $S[0,\pi]$, then $q$ is a polynomial.
\end{lem}

\begin{proof}
    Since the pole set of the meromorphic function $F$ has no finite accumulation points, and its intersection with $\mathbb{S}[0,2\pi]$ must be located in a compact set due to \ref{fperiod} and the estimate \eqref{eq:fasest}, this intersection must be finite.  Now, define $\zeta \mapsto G(\zeta) := F(\zeta)/F_{S,\mathrm{min}}(\zeta)$ which satisfies $G(\zeta) = G(-\zeta) = G(\zeta +2\pi i)$. Then, analogous to the proof of Lemma~\ref{lem:minsol}, there exists a meromorphic function $q = G\circ \ch^{-1}$ which is polynomially bounded at infinity and has finitely many poles. Thus, it is a rational function. Lastly, note that $q(-1)=1$ due to $G(i\pi)=1$.
\end{proof}

\subsection{Existence of the minimal solution and its asymptotic growth}\label{app:fminexist}
In this section, we establish the existence of a common integral representation of the minimal solution for a large class of (eigenvalues of) regular S-functions, namely those satisfying the hypothesis of Theorem~\ref{thm:propsofintreps} below. As a byproduct, but of crucial importance for our discussion in Section~\ref{sec:examples}, we obtain an explicit formula for the asymptotic growth of the minimal solution (Proposition~\ref{prop:minasympt}).

For $\mathbb{C}$-valued functions $S$ and $f$, the integral expressions of interest are formally given by
\begin{align}
    f[S](t) & := \frac{i}{\pi} \int_0^\infty S'(\theta)S(\theta)^{-1} \cos (\pi^{-1} \theta t) d\theta, \label{app:eq:fdef}\\
    S_f(\zeta) & := \exp \left(-2i\int_0^\infty f(t) \sin \frac{\zeta t}{\pi} \,\frac{dt}{t}\right) , \label{app:eq:sint}\\
    F_f(\zeta) & := \exp \left(2\int_0^\infty f(t) \sin^2 \frac{(i\pi-\zeta) t}{2\pi} \, \frac{dt}{t\sh t}\right); \label{app:eq:fint}
\end{align}
we will give conditions for their well-definedness below. $f[S]$ will be referred to as the \emph{characteristic function}\footnote{Differing conventions for $f[S]$ are found in the literature. In the form factor programme community, one mostly takes $2f[S]$ as the characteristic function: Compare formulas \eqref{app:eq:sint}--\eqref{app:eq:fint} with, e.g., \cite[Eq.~\lParent 4.10\rParent --\lParent 4.11\rParent ]{FMS93} or \cite[Eq.~\lParent 2.18\rParent --\lParent 2.19\rParent ]{KW78}, but noting a typo in Eq.~(2.19) there.}
with respect to $S$. 
For a large class of functions $S$, the functions $S_{f[S]}$ and $F_{f[S]}$ will agree with $S$ and $F_{S,\mathrm{min}}$ respectively.

For the following let us agree to call a function $f$ on $\mathbb{R}$ \emph{exponentially decaying} iff
\begin{equation}\label{eq:expdecay1}
    \exists a,b,r > 0 \, \forall |t| \geq r : \quad |f(t)| \leq a \exp (-b |t|);
\end{equation}
analogously for functions on $[0,\infty)$. A function $f$ on a strip $\mathbb{S}(-\epsilon,\epsilon)$ will be called \emph{uniformly $L^1$} if $f(\cdot + i \lambda) \in L^1(\mathbb{R})$ for every $\lambda\in(-\epsilon,\epsilon)$, with the $L^1$ norm uniformly bounded in $\lambda$.

Now, we are ready to state the main result:
\begin{thm}\label{thm:propsofintreps}
    Let $S:\mathbb{C}\to\mathbb{C}$ be a meromorphic function with no poles on the real line, satisfying $S(\zeta)^{-1} = S(-\zeta)$, and regularity \ref{sreg}. Suppose that $r_S(\zeta):=iS'(\zeta)/S(\zeta)$ is uniformly $L^1$ on some strip $\mathbb{S}(-\epsilon,\epsilon)$. Then $f[S] \in C([0,\infty),\mathbb{R})$ is exponentially decaying. If further $f[S] \in C^2([0,\delta))$ for some $\delta > 0$, then the minimal solution with respect to $S$ exists.

    In more detail, under these assumptions $S_{f[S]}$ and $F_{f[S]}$ are well-defined, nonvanishing and analytic on $\mathbb{S}(-\epsilon,\epsilon)$ and $\mathbb{S}(-\epsilon,2\pi+\epsilon)$, respectively. The meromorphic continuations of $S_{f[S]}$ and $F_{f[S]}$ to all of $\mathbb{C}$ exist. In case that $S(0) = 1$, we have $S_{f[S]}=S$ and $F_{f[S]}=F_{S,\mathrm{min}}$. In case that $S(0) = -1$, we have $S_{f[S]}=-S$ and $F_{f[S]}=F_{-S,\mathrm{min}}$; and $F_{S,\mathrm{min}}(\zeta) = (-i \sh \tfrac{\zeta}{2}) F_{-S,\mathrm{min}}(\zeta)$.
\end{thm}

The examples treated in Section~\ref{sec:examples} fulfil this condition: In particular, for products of S-functions of sinh-Gordon type (see Eq.~\eqref{eq:prodshG}), $r_S$ is exponentially decaying (on a strip) and $f[S]$ is actually smooth on $[0,\infty)$ (cf. Eq.~\eqref{eq:shgfmin}). For the eigenvalues of the S-function of the $O(n)$-NLS model, $S \in \{ S_0, S_+, S_-\}$, we find $r_{S}(\theta+i\lambda) \lesssim \theta^{-2}$, $|\theta|\to \infty$, uniformly in $\lambda \in [-\epsilon,\epsilon]$ for some $\epsilon > 0$ and again, that $f[S]$ is smooth on $[0,\infty)$ (cf. Appendix~\ref{app:charfct}).

We give the proof of the theorem in several steps. To begin with, we have:

\begin{lem}\label{lem:charfctgrowth}
    Let $r_S$ be uniformly $L^1$ on a strip $\mathbb{S}(-\epsilon,\epsilon)$. Then $f[S]:\mathbb{R} \to \mathbb{R}$ is an even, bounded, and continuous function which decays exponentially.
\end{lem}

\begin{proof}
    Since $r_S$ restricted to $\mathbb{R}$ is $L^1$-integrable, its Fourier transform $\widetilde{r_S}$ is bounded, continuous and vanishes towards infinity. As $r_S$ is even, 
    \begin{equation}r_S(-\theta) = iS'(-\theta)S(\theta) = -i(S(\theta)^{-1})'S(\theta) = iS'(\theta) S(\theta)^{-1} = r_S(\theta),\end{equation}
    also $\widetilde{r_S}$ is even and we have
    \begin{equation}
        f[S](t) = \tfrac{i}{\pi}\int_0^\infty r_S(\theta) \cos(\pi^{-1}\theta t) dt = \tfrac{i}{2\pi} \int r_S(\theta) e^{i\pi^{-1} \theta t} dt = \tfrac{i}{2\pi} \widetilde{r_S}(t).
    \end{equation}
    Let now $0 < \lambda < \epsilon$ be arbitrary. By assumption, $r_S(\cdot+i\lambda)$ is $L^1$-integrable as well; and by the translation property of the Fourier transform,
    \begin{equation}
        \widetilde{r_S}(t) = e^{-\lambda t} \tfrac{1}{2\pi} \widetilde{r_S(\cdot+i\lambda)}(t),
    \end{equation}
    where $\widetilde{r_S(\cdot+i\lambda)}(t)$ vanishes for $|t|\to \infty$ due to the Riemann-Lebesgue lemma. Thus $\widetilde{r_S}$ is exponentially decaying towards $+\infty$, and since it is even also towards $-\infty$.
\end{proof}

We continue with an existence result for $S_f$ and $F_f$ and some of their basic properties:

\begin{lem}\label{lem:fminwelldef}
    Let $f \in C([0,\infty),\mathbb{R})$ be exponentially decaying. Then the functions $S_f$ and $F_f$ are well-defined by the integral expressions \eqref{app:eq:sint} and \eqref{app:eq:fint}. Further they are nonvanishing and analytic on $\mathbb{S}(-\epsilon,\epsilon)$ and $\mathbb{S}(-\epsilon,2\pi + \epsilon)$, respectively, for some $\epsilon > 0$.
\end{lem}

\begin{proof}
    By assumption there exist positive constants $a,r,\epsilon > 0$ such that $|f(t)| < a \exp(-\epsilon t)$ for all $t \geq r$. By the triangle inequality, $|\sin \pi^{-1} \zeta t| = \tfrac{1}{2} | e^{i\pi^{-1}\zeta t} + e^{-i\pi^{-1}\zeta t}| \leq \exp (\pi^{-1} |\Im\zeta| t)$ for all $t\geq 0$. Thus we find that $f(t) t^{-1} \sin (\pi^{-1} \zeta t)$ is exponentially decaying for $|\Im \zeta| < \epsilon$. Also, this function is continuous (including at $t=0$ because of the first-order zero of the sine function at zero). By similar arguments, the same holds for its derivative with respect to $\zeta$. In particular, $t \mapsto f(t) t^{-1} \sin (\pi^{-1} \zeta t)$ and its $\zeta$-derivative are absolutely integrable for all $\zeta \in \mathbb{S}(-\epsilon,\epsilon)$. As a consequence, $S_f$ is well-defined and analytic on $\mathbb{S}(-\epsilon,\epsilon)$. Since $S_f$ is given by an exponential, it does not vanish.
    
    The argument for $F_f$ runs analogously. The estimate from above implies that $|\sin^2 (2\pi)^{-1} (i\pi-\zeta) t| \leq \exp(\pi^{-1} |\Im(i\pi-\zeta)| t)$ for all $t\geq 0$. As a consequence, $t\mapsto f(t) (t\sh t)^{-1} \sin^2 ((2\pi)^{-1} (i\pi-\zeta) t)$ is exponentially decaying for $|\Im \zeta - \pi| < \pi+\epsilon$. It is further continuous (including at $t=0$ because of the second-order zero of the sine-function at zero). Together with similar properties of the $\zeta$-derivative, it follows that $F_f$ is well-defined, analytic, and nonvanishing in the region $\mathbb{S}(-\epsilon,2\pi + \epsilon)$.
\end{proof}

\begin{lem}\label{lem:minsolprops}
    Let $f \in C([0,\infty),\mathbb{R})$ be exponentially decaying. Then $S_f(0) = 1$ and $F_f(i\pi)=1$. Moreover, there exists $\epsilon > 0$ such that for all $\zeta \in \mathbb{S}(-\epsilon,\epsilon)$,
    \begin{equation}
        S_f(\zeta)^{-1} = S_f(-\zeta)
    \end{equation}
    and
    \begin{equation}\label{eq:ffrel}
        F_f(\zeta) = S_f(\zeta) F_f(-\zeta), \quad F_f(i\pi + \zeta) = F_f(i\pi -\zeta).
    \end{equation}
\end{lem}

\begin{proof}
    $S_f(0)= 1$ and $F_f(i\pi)=1$ is immediate by definition. Next, take $\zeta \in \mathbb{S}(-\epsilon,\epsilon)$ with the $\epsilon$ from Lemma~\ref{lem:fminwelldef}. Then 
    \begin{equation}-\sin( \pi^{-1} \zeta t) = \sin( \pi^{-1} (-\zeta) t)\end{equation} 
    implies that $S_f(\zeta)^{-1} = S_f(-\zeta)$. Similarly,
    \begin{equation}
       \sin^2 \frac{\zeta t}{2\pi} = \sin^2 \frac{-\zeta t}{2\pi} 
    \end{equation} 
    implies that $F_f(i\pi + \zeta) = F_f(i\pi -\zeta)$. Lastly, the relation
    \begin{equation}
    \sin^2 \frac{(i\pi - \zeta) t}{2\pi} - \sin^2 \frac{(i\pi+\zeta)t}{2\pi} = -i \sh (t) \sin \frac{\zeta t}{\pi}
    \end{equation}
    implies that    
    \begin{equation}\begin{aligned}
        \log \frac{F_f(\zeta)}{F_f(-\zeta)} & = 2\int_0^\infty f(t) \left( \sin^2 \frac{(i\pi-\zeta) t}{2\pi} - \sin^2 \frac{(i\pi + \zeta)t}{2\pi}\right) \, \frac{dt}{t\sh t} \\
        & = -2i \int_0^\infty f(t) \sin(\pi^{-1} \zeta t) \, \frac{dt}{t} = \log S_f(\zeta),
    \end{aligned}\end{equation}
    which concludes the proof.
\end{proof}

\begin{cor}\label{cor:charfct}
    For arbitrary $c \in \mathbb{R}$ and exponentially decaying $f,g\in C([0,\infty),\mathbb{R})$, one has
    \begin{enumerate}[label=(\alph*)]
     \item $S_{c f} = (S_f)^c, \quad F_{c f} = (F_f)^c,$
     \item $S_{f+g} = S_f S_g, \quad F_{f+g} = F_f F_g,$
     \item \label{cor:charfctinject} $S_f = S_g \Leftrightarrow f = g, \quad F_f = F_g \Leftrightarrow f = g.$
    \end{enumerate}
\end{cor}

\begin{proof}
    (a) and (b) are immediate since $\log S_f$ and $\log F_f$ are linear in $f$ by definition. 
    In (c), we only need to show ``$\Rightarrow$'', and by the previous parts, we may assume $g=0$, with $S_0 = F_0 = 1$. 
    Now if $S_f = 1$, we compute from Eq.~\eqref{app:eq:sint} for any $\lambda \in \mathbb{R}$,
    \begin{equation}
        0 = \frac{d}{d\lambda} \log S_f(\lambda) 
        =  -2i \frac{d}{d\lambda} \int_0^\infty f(\pi t) \sin (\lambda t) \frac{dt}{t}
        = -2i \int_0^\infty f(\pi t) \cos (\lambda t) \,dt,
    \end{equation}
    hence $f=0$ by the inversion formula for the Fourier cosine transform.
    
    If $F_f = 1$, we use \eqref{eq:ffrel} to conclude that $S_f = 1$, which implies $f=0$ as seen earlier. 
\end{proof}

\begin{prop}[Asymptotic estimate]\label{prop:minasympt}
    Let $f \in C([0,\infty),\mathbb{R})$ be exponentially decaying and $C^2([0,\delta))$ for some $\delta > 0$. Let $f_0:=f(0)$ and $f_1 := f'(0)$, where $^\prime$ refers to the half-sided derivative. Then there exist constants $0 < c \leq c'$ and $r>0$ such that
    \begin{equation}\label{eq:minasympt}
        \forall |\Re \zeta|\geq r, \Im \zeta \in [0,2\pi]: \quad c \leq \frac{|F_f(\zeta)|}{|\Re \zeta|^{f_1} \exp |\Re \zeta|^{f_0/2}} \leq c'.
    \end{equation}
\end{prop}
\begin{proof}
    In the following, let $z:= (i\pi-\zeta)/2\pi $ with $x:=|\Re z|>0$ and $y:=|\Im z| \leq \frac{1}{2}$. Then 
    \begin{align}\label{def:iz}
        \Re \log F_f(\zeta) & = 2\int_0^\infty \frac{f(t)}{t \sh t} \Re \sin^2(zt ) dt \non \\
            & = \int_0^\infty \frac{f(t)}{t \sh t} (1-\cos 2xt \ch 2yt) dt =: I(z).
    \end{align}
    The aim is to show that $|I(z) - f_0 \pi x - f_1 \log x|$ is uniformly bounded in $z\in \mathbb{S}[-\tfrac{1}{2},\tfrac{1}{2}]$,
    as this implies the bound \eqref{eq:minasympt} by monotonicity of the exponential function.
    To begin with, note that the integrand of \eqref{def:iz} for $t\geq 1, y\leq \frac{1}{2}$, is uniformly bounded by $f(t)(t \sh t)^{-1}(1+ \ch t)$. This is integrable on $[1,\infty)$ by the exponential decay of $f$. The integral over $[1,\infty)$ is thus bounded uniformly in $z$ by a constant $c_0$.
    
    As further preliminaries, let us note that
    \begin{align}
        \left|\int_0^1  ((t\sh t)^{-1} - t^{-2}) (1-\cos 2xt \ch 2yt) f(t) dt\right| & \leq c_1, \label{eq:pre1} \\
        \left|\int_0^1 ( \ch (2y t) - 1)t^{-2} \cos 2xt f(t) dt \right| & \leq c_2, \label{eq:pre2}
    \end{align}
    where $c_1$, $c_2$ are constants independent of $x$ and $y$. This is implied by mean-value-estimates using regularity of the functions $(t\sh t)^{-1} - t^{-2}$ and $(\ch(2yt) -1)t^{-2}$ also at $t=0$, where $t$ and $y$ are evaluated on compact ranges while $x$ appears only in the argument of the cosine-function. The same reasoning allows us to estimate
    \begin{equation}
        \left|\int_0^1 (f(t)-f_0-f_1t)t^{-2} (1-\cos 2xt) dt \right| \leq c_3, \label{eq:pre3}
    \end{equation}
    where we apply Taylor's theorem to $f\in C^2([0,\delta))$ to argue regularity at $t=0$.        
    
    Applying Eqs.~\eqref{eq:pre1}-\eqref{eq:pre3} and the triangle inequality yields
    \begin{equation}
        \left| I(z) -  J(x) \right| \leq c_0+c_1+c_2+c_3,
    \end{equation}
    where
    \begin{equation}
      J(x) := \int_0^{1} (f_0+f_1 t) \frac{1-\cos 2xt}{t^2} dt =  
      f_0 \big(-1 + \cos 2x + 2x \operatorname{Si} (2x)\big) + f_1 \operatorname{Cin}(2x)
    \end{equation}
    in terms of the standard sine and cosine integral functions. Since these have the asymptotics $\operatorname{Si}(x) = \frac{\pi}{2} + \mathcal{O}(x)$, $\operatorname{Cin}(x) = \log x + \mathcal{O}(1)$ as $x\to\infty$ \cite[\S{}6.12(ii)]{NIST:DLMF}, one finds constants $r,c>0$ such that
    \begin{equation}
         \forall x \geq r: \quad \left| I(z) - f_0 \pi x - f_1 \log x \right| \leq c. \qedhere
    \end{equation}
\end{proof}

With the asymptotic estimate for $F_f$ we can now prove the main result:

\begin{proof}[Proof of Theorem~\ref{thm:propsofintreps}]
    First, consider $S(0)=1$. By Lemma~\ref{lem:charfctgrowth}, $f[S]$ is exponentially decaying and hence by Lemma~\ref{lem:fminwelldef}, $S_{f[S]}$ is well-defined, and analytic and nonvanishing on a small strip. Combining \eqref{app:eq:fdef} and \eqref{app:eq:sint} with the inversion formula for the Fourier cosine transform, we find for $\lambda\in\mathbb{R}$,
    \begin{equation}
        \frac{d}{d\lambda } \log S_{f[S]}(\lambda) = -\frac{2i}{\pi} \int_0^\infty f[S](t) \cos\frac{\lambda t}{\pi}  dt = \frac{S'(\lambda)}{S(\lambda)}
        =\frac{d}{d\lambda } \log S(\lambda).
    \end{equation}
    Since also $S_{f[S]}(0)=1=S(0)$, we conclude that $S_{f[S]}=S$ first on the real line, and then as meromorphic functions.
    
    Further by Lemma~\ref{lem:minsolprops}, $F:=F_{f[S]}$ is analytic and nonvanishing on the physical strip $\mathbb{S}[0,\pi]$, satisfies $F(i\pi) = 1$, and for some $\epsilon > 0$,
    \begin{equation}\label{eq:frels}
        F(\zeta) = S(\zeta) F(-\zeta), \quad F(i\pi + \zeta) = F(i\pi - \zeta), \quad \zeta \in \mathbb{S}(-\epsilon,\epsilon);
    \end{equation}
    in fact, using these relations we can extend $F$ as a meromorphic function to all of $\mathbb{C}$. Also, Proposition~\ref{prop:minasympt} yields the asymptotic estimate in Lemma~\ref{lem:minsol}\ref{fest}. In summary, Lemma~\ref{lem:minsol} applies to $F$; hence, $F$ is the unique minimal solution with respect to $S$.
    
    In the case $S(0)=-1$, one finds $S_{f[S]} = S_{f[-S]} = -S$ by the above; also, $F_{f[S]}=F_{f[-S]}$ is the minimal solution with respect to $-S$, and from Corollary~\ref{cor:prodminsol}, we have $F_{S,\text{min}}(\zeta) = -i\sh \tfrac{\zeta}{2}  F_{f[-S]} (\zeta)$.
\end{proof}

\newpage

\section{Computing a characteristic function}\label{app:charfct}

In this appendix, we present a method to explicitly compute characteristic functions (as defined in Appendix \ref{app:fminexist}) for a certain class of $S$. The method is known but only briefly described in \cite{KTTW77}. We illustrate it here using the eigenvalues of the S-function of the $O(n)$-nonlinear sigma model, i.e., $S_i$ for $i=\pm,0$ (see Definition~\ref{defn:onnls} and below). First, we present the general method; second, we check the examples $f[S_\pm]$ against the literature; lastly, we compute $f[S_0]$.

The method applies to S-function eigenvalues which are given as a product of Gamma functions; see \cite[Appendix C]{BFKZ99} for some typical examples. While this product can be infinite in general, we restrict here to finite products, which suffice for our purposes. Specifically, let $S$ be of the form
\begin{equation}\label{eq:sformgamma}
    S(\theta) = \frac{\prod_{x\in A_+} \Gamma(x+\frac{\theta}{\lambda \pi i}) \prod_{x\in A_-} \Gamma(x-\frac{\theta}{\lambda \pi i})}{\prod_{x\in A_+} \Gamma(x-\frac{\theta}{\lambda \pi i})\prod_{x\in A_-} \Gamma(x+\frac{\theta}{\lambda \pi i})},
\end{equation}
where $\lambda > 0$ and $A_\pm$ are finite subsets of $(0,\infty)$ such that $|A_+|=|A_-|$. It is straightforward to check that this indeed defines a regular $\mathbb{C}$-valued S-function, apart from crossing symmetry which can only be satisfied for $A_+=A_-=\emptyset$ or infinite products.

\begin{lem}\label{lem:computecharfct}
     The characteristic function with respect to $S$ as in Eq.~\eqref{eq:sformgamma} is
    \begin{equation}\label{eq:computecharfctresult}
        t\mapsto f[S](t) = \frac{1}{1-e^{-\lambda t}} \left(\sum_{x\in A_-}-\sum_{x\in A_+} \right) e^{-\lambda xt}.
    \end{equation}
\end{lem}

\begin{proof}
    Since $f[S]$ is linear in $\log S$, it suffices to consider the case $A_+=\{x_+\}$, $A_-=\{x_-\}$. Using Malmst\'en's formula (see, e.g., \cite[Sec.~1.9]{Bat53})
    \begin{equation}
        \log \Gamma (z) = \int_0^\infty \left( z - 1 - \frac{1-e^{-(z-1)t}}{1-e^{-t}}\right) \frac{e^{-t}}{t} dt, \quad \Re z > 0,
    \end{equation}
    we find
    \begin{align}
        \frac{d}{d\theta} \log S(\theta) & = \int_0^\infty \frac{\big(e^{\frac{\theta  t}{\lambda \pi i}} + e^{\frac{-\theta t}{\lambda \pi i}}\big)\big(e^{-x_+ t} - e^{-x_- t}\big)}{1-e^{-t}} \frac{dt}{\lambda \pi i} \\
        &= -\tfrac{2i}{\pi} \int_0^\infty \underbrace{\frac{e^{-\lambda x_- t}-e^{-\lambda x_+t}}{1-e^{-\lambda t}}}_{=:g(t)}
        \cos \frac{\theta t}{\pi} dt.
    \end{align}
    By definition in \eqref{app:eq:fdef}, $f[S]$ is given as the Fourier cosine transform of $S(\theta)^{-1} \frac{d}{d\theta} S(\theta) = \frac{d}{d\theta} \log S(\theta)$; its inversion formula yields that $f[S]=g$ since $g$ is clearly integrable.
\end{proof}

\begin{ex}[Eigenvalues $S_\pm$]
    By definition, $S_\pm(\theta) = (b \pm c)(\theta) = h_\pm(\theta) b(\theta)$ with
    \begin{equation}b(\theta) = s(\theta) s(i\pi -\theta), \quad s(\theta) = \frac{\Gamma\left(\frac{\nu}{2} + \frac{\theta}{2\pi i}\right)\Gamma\left(\frac{1}{2} + \frac{\theta}{2\pi i}\right)}{\Gamma\left(\frac{1+\nu}{2} +\frac{\theta}{2\pi i}\right) \Gamma\left(\frac{\theta}{2\pi i}\right)} \end{equation}
    and
    \begin{equation}h_\pm(\theta) = \frac{\theta \mp i\pi \nu}{\theta} = \mp \frac{\frac{\nu}{2} \mp \frac{\theta}{2\pi i}}{\frac{\theta}{2\pi i}} = \mp \frac{\Gamma(1+\tfrac{\nu}{2}\mp \tfrac{\theta}{2\pi i})\Gamma(\tfrac{\theta}{2\pi i})}{\Gamma(\tfrac{\nu}{2} \mp \tfrac{\theta}{2\pi i})\Gamma(1+\tfrac{\theta}{2\pi i})},\end{equation}
    where we used $z = \Gamma(z+1)/\Gamma(z)$ in order to represent $h_\pm$ in terms of $\Gamma$.
    As a result, 
    \begin{equation}S_\pm(\theta) = \mp \frac{\Gamma(\tfrac{1\mp 1}{2} + \frac{\nu}{2}+\frac{\theta}{2\pi i} ) \Gamma(\frac{1}{2}+\frac{\theta}{2\pi i} ) \Gamma(\frac{1}{2} + \frac{\nu}{2}-\frac{\theta}{2\pi i} ) \Gamma(1-\frac{\theta}{2\pi i} )}{\Gamma(\frac{1}{2}+\frac{\nu}{2}+\frac{\theta}{2\pi i} ) \Gamma(1+\frac{\theta}{2\pi i} ) \Gamma(\tfrac{1\mp 1}{2} +\frac{\nu}{2}-\frac{\theta}{2\pi i} ) \Gamma(\frac{1}{2}-\frac{\theta}{2\pi i} )},\end{equation}
    which is of the form \eqref{eq:sformgamma} for $\lambda = 2$, $A_+ = \{ \frac{1}{2}, \tfrac{1\mp 1}{2} + \frac{\nu}{2} \}$, and $A_- = \{ 1,  \frac{1}{2}+\frac{\nu}{2} \}$. Due to Lemma~\ref{lem:computecharfct} we find
    \begin{align}\label{eq:fsp}
        f[-S_+](t) & = \frac{1}{1-e^{-2t}} \left( e^{-2t} + e^{-(\nu+1)t} - e^{-t} - e^{-\nu t}  \right) = -\frac{1 + e^{(1-\nu) t}}{e^t+1}, \\
        f[S_-](t) & = \frac{1}{1-e^{-2t}} \left( e^{-2t} + e^{-(\nu+1)t} - e^{-t} - e^{-(\nu + 2)t} \right) = \frac{e^{-\nu t}-1}{e^{t}+1}.\label{eq:fsm}
    \end{align}
    This agrees with \cite[Eq.~(2.11)]{BFK13} and \cite[Eq.~(5.7)]{KW78}. We read off
    \begin{align}
        f[-S_+](t) & = -1 + \tfrac{\nu}{2} t + \mathcal{O}(t^2), \quad t\to 0; \\
        f[S_-](t) & = -\tfrac{\nu}{2} t + \mathcal{O}(t^2), \quad t \to 0.
    \end{align}
\end{ex}
\begin{ex}[Eigenvalue $-S_0$]
    Similarly to the preceding example, we write
    \begin{equation}S_0(\theta) = h_0(\theta) b(\theta)\end{equation}
    with 
    \begin{equation}h_0(\theta) = \frac{\theta^2 + i\pi (1+\nu) \theta -\nu \pi^2 }{\theta (\theta-i\pi)} = - \frac{(\frac{1}{2} + \frac{\theta}{2\pi i})(\frac{\nu}{2}+\frac{\theta}{2\pi i})}{\frac{\theta}{2\pi i} (\frac{1}{2}-\frac{\theta}{2\pi i})}.\end{equation}
    Using again $z = \Gamma(z+1)/\Gamma(z)$, we find
    \begin{equation}S_0(\theta) = -\frac{\Gamma(1+\frac{\nu}{2}+\frac{\theta}{2\pi i} ) \Gamma(\frac{3}{2}+\frac{\theta}{2\pi i} ) \Gamma(\frac{1}{2} + \frac{\nu}{2}-\frac{\theta}{2\pi i} ) \Gamma(1-\frac{\theta}{2\pi i} )}{\Gamma(\frac{1}{2}+\frac{\nu}{2}+\frac{\theta}{2\pi i} ) \Gamma(1+\frac{\theta}{2\pi i} ) \Gamma(1 +\frac{\nu}{2}-\frac{\theta}{2\pi i} ) \Gamma(\frac{3}{2}-\frac{\theta}{2\pi i} )},\end{equation}
    which is of the form \eqref{eq:sformgamma} for $\lambda = 2$, $A_+ = \{\frac{3}{2}, 1 + \frac{\nu}{2} \}$, and $A_- = \{1, \frac{1}{2} + \frac{\nu}{2} \}$. Due to Lemma~\ref{lem:computecharfct} we find
    \begin{equation}\label{eq:fs0}
        f[-S_0](t) = \frac{1}{1-e^{-2t}} \left( e^{-2t} + e^{-(1+\nu)t}-e^{-3t} - e^{-(2+\nu)t} \right) = \frac{e^{-t} + e^{-\nu t}}{e^t + 1}
    \end{equation}
    and conclude
    \begin{equation}f[-S_0](t) = 1 -(1+\tfrac{\nu}{2}) t + \mathcal{O}(t^2), \quad t\to 0.\end{equation}
\end{ex}

\vspace{5cm}

\section*{Acknowledgements}

J.M.\ would like to thank Karim Shedid Attifa and Markus Fr\"ob for many fruitful discussions. D.C.\ and J.M.\ acknowledge support by the Deutsche Forschungsgemeinschaft (DFG) within the Emmy Noether Grant
CA1850/1-1 and Grant No. 406116891 within the Research Training Group RTG 2522/1. H.B.\ would like to thank the Institute for Theoretical Physics at the University of Leipzig for
hospitality.

\newpage

\addcontentsline{toc}{section}{References}
\renewcommand{\mkbibnamefamily}[1]{\textsc{#1}}

\printbibliography

\end{document}